\documentclass[runningheads,orivec]{llncs}
\usepackage{graphicx}

\usepackage[utf8]{inputenc}

\usepackage{float}
\usepackage{tikz} 
\usetikzlibrary{shapes}
\usetikzlibrary{automata}
\usetikzlibrary{calc}
\usepackage{amssymb}
\usepackage{mathpartir}
\usepackage{comment} 
\usepackage{multicol}
\usepackage{scalerel}
\usepackage{bussproofs}
\usepackage{lscape}
\usepackage{stmaryrd}
\usepackage[title]{appendix}
\usepackage{ upgreek }
\usepackage{listings}
\usepackage{stackengine}
\usepackage[tworuled, vlined ]{algorithm2e}
\usepackage{macro}
\usepackage{capt-of}
\usepackage{chngcntr}
\usepackage{caption}
\usepackage{mathtools}

\begin{document}
\title{On the k-synchronizability of systems}
%
%
\author{Cinzia Di Giusto\inst{1} \and
Laetitia Laversa\inst{1} \and
Etienne Lozes\inst{1}}
\authorrunning{C. Di Giusto et al.}
%
\institute{Université Côte d’Azur, CNRS, I3S, France}
%
\maketitle              
\begin{abstract}

In this paper, we work on the notion of  \kSity{k}: a system is \kSable{k}
if any of its executions, up to reordering causally
independent actions, can be divided into a succession of
 $k$-bounded interaction phases. We show two results (both for mailbox and peer-to-peer automata):
 first, the reachability problem  is decidable for
 \kSable{k} systems; second, the membership problem (whether a
 given system is \kSable{k}) is decidable as well.
 Our proofs fix several important issues in previous attempts
 to prove these two results for mailbox automata.

\end{abstract}


\section{Introduction}
Asynchronous message-passing is ubiquitous in communication-centric systems;
these include high-performance computing, distributed
memory management, event-driven programming, or web services
orchestration. 
One of the parameters that play an important role in these systems
is whether 
the number of pending sent messages can be
bounded in a predictable fashion, or 
whether the buffering capacity offered by the communication layer should be unlimited.
Clearly, when considering implementation, testing, or verification,
bounded asynchrony is  preferred over unbounded asynchrony.
Indeed, for bounded systems, 
reachability analysis and invariants inference can be solved by
regular model-checking~\cite{DBLP:conf/cav/BouajjaniHV04}. 
Unfortunately and even if designing a new system in this setting is easier, this is not the case when considering that the buffering capacity is unbounded, or that the bound is not
known a priori . 
Thus, a  question that arises naturally is  how can we bound the ``behaviour" of  a  system so that it operates as one 
with unbounded buffers?
 In a recent work~\cite{DBLP:conf/cav/BouajjaniEJQ18}, Bouajjani~\emph{et al.}
introduced the notion of \kSable{k} system of finite
state machines communicating through mailboxes and showed that the reachability problem is decidable for such systems.
Intuitively, 
 a system is \kSable{k}
if any of its executions, up to reordering causally
independent actions, can be chopped into a succession of
$k$-bounded interaction phases.
Each of these
phases starts with at most $k$ send actions that are 
followed by at most $k$ receptions. 
Notice that,  a system may be \kSable{k} even if
some of its executions require buffers of unbounded capacity. 

%


As explained in the present paper, this result, although valid, 
is surprisingly non-trivial, 
mostly due to complications introduced by the mailbox semantics of
communications. Some of these complications were 
missed by Bouajja\-ni~\emph{et al.} and the algorithm for the
reachability problem in~\cite{DBLP:conf/cav/BouajjaniEJQ18}
suffers from false positives. 
Another problem is the membership problem for the subclass of \kSable{k}
systems: for a given $k$ and a given system of communicating finite state
machines, is this system \kSable{k}?
The main result in~\cite{DBLP:conf/cav/BouajjaniEJQ18}
is that this problem is decidable. However, again,
the proof of this result 
contains an important flaw
at the very first step that breaks all subsequent developments; as a consequence, the
algorithm given in~\cite{DBLP:conf/cav/BouajjaniEJQ18} produces both false positives and false negatives.

In this work, we present a new proof of the decidability of the reachability
problem together with a new proof of the decidability of
the membership pro\-blem. Quite surprisingly, the reachability problem is more demanding in terms of causality analysis,
whereas the membership problem, although rather intricate, builds on a simpler dependency analysis.
We also extend both decidability results to the case of peer-to-peer communication.

\vspace{0.1cm}
\noindent \textbf{Outline.} Next section
recalls the definition of communicating systems and related notions.
In Section~\ref{sec:ksity} we introduce \kSity{k}
and we give a gra\-phi\-cal characterisation of this property.
This characterisation corrects Theorem~1
in~\cite{DBLP:conf/cav/BouajjaniEJQ18} and highlights the flaw in the proof
of the membership problem. Next, in Section~\ref{sec:reachability},
we establish the decidability of the reachability problem, which
is the core of our contribution and departs considerably from~\cite{DBLP:conf/cav/BouajjaniEJQ18}. In Section~\ref{sec:dec}, we show the decidability of the membership problem.
Section \ref{section:p2p} extends previous results to the peer-to-peer setting.
Finally Section \ref{sec:conc} concludes the paper discussing other related works.
Proofs and some additional material are added in a separate  Appendix.


\section{Preliminaries}

A  communicating system is a set of finite state machines that exchange messages:  automata have  transitions labelled with either send or receive actions.
The paper mainly  considers as communication architecture, mailboxes: \ie messages await to be
received in  FIFO buffers that store  all
messages sent to a same automaton, regardless of their senders. 
Section \ref{section:p2p}, instead, treats peer-to-peer systems, their introduction is therefore delayed to that point.

Let $\paylodSet$ be a finite set of messages and $\procSet$ a finite set of processes.  A send action, denoted $\send{p}{q}{\amessage}$, designates the sending of message $\amessage$ from process $p$ to process $q$. Similarly a receive action $\rec{p}{q}{\amessage}$ expresses that process $q$  is receiving  message
$\amessage$ from $p$.  We write $a$ to denote a send or receive action.
Let
$\sendSet~=~\{ \send{p}{q}{\amessage} \mid  p,q \in \procSet, \amessage \in \paylodSet \}$ be
the set of send actions and $\receiveSet~=~\{ \rec{p}{q}{\amessage}\mid  p,q \in \procSet, \amessage \in \paylodSet\}$ the set of receive actions. $\sendSet_p$ and $\receiveSet_p$ stand for the set of sends and   receives  of process $p$ respectively.
Each process is encoded by an automaton and by abuse of notation we say that a \emph{system} is the parallel composition of 
processes.

\begin{definition}[System]
  A system is a tuple $ \system = \left( (L_p, \delta_p, l^0_p) \mid p \in \procSet \right) $ where, for each process $p$, $L_p$ is a finite set of local
  control states, $\delta_p \subseteq (L_p \times ( \sendSet_p \cup \receiveSet_p ) \times L_p)$ is the transition relation (also denoted $l \xrightarrow{a}_p l'$) and $l^0_p$ is the initial state.
\end{definition}

\begin{definition}[Configuration]
  Let $\system= \left( (L_p, \delta_p, l^0_p) \mid p \in \procSet \right) $, a configuration is a pair $(\globalstate{l}, \B)$
  where $\globalstate{l}=(l_p)_{p\in\procSet} \in \Pi_{p \in \procSet} L_p$
  is a global control state of $\system$ (a local control state for each automaton), and $\B= (b_p)_{ p \in \procSet} \in (\paylodSet^*)^\procSet$ is a vector of buffers,
  each $b_p$ being a word over $\paylodSet$.
\end{definition}
We write
$\globalstate{l_0}$ to denote
the vector of initial states of all processes
$p \in \procSet$, and $\B_0$ stands for the vector
of empty  buffers.
The semantics of a system is defined by the two rules below.

\begin{center}
\begin{minipage}{.5\columnwidth}
\small{[SEND]}

\AxiomC{{\small $ \Transition{l_p}
               {\send{p}{q}{\amessage}}
               {l_p'}
               {p}
    \quad
    b_{q}' = b_{q} \cdot \amessage$}}
\UnaryInfC{{\small $\Transition{(\globalstate{l}, \B)}
               {\send{p}{q}{\amessage}}
	       {(\globalstate{l}\sub{l_p'}{l_p},\B\sub{b_{q}'}{b_{q}})}
               {}$}}
\DisplayProof
\end{minipage}
\begin{minipage}{.5\columnwidth}
\small{[RECEIVE]}
\AxiomC{{\small$ \Transition{l_q}
               {\rec{p}{q}{\amessage}}
               {l_q'}
               {q}
    \quad
    b_{q} = \amessage \cdot b_{q}'$}}
\UnaryInfC{{\small$\Transition{(\globalstate{l}, \B)}
               {\rec{p}{q}{\amessage}}
	       {(\globalstate{l}\sub{l_q'}{l_q},\B\sub{b_{q}'}{b_{q}})}
               {}$}}
\DisplayProof
\end{minipage}
\end{center}

%
\noindent A send action adds a message in the buffer $b$ of the receiver, and
a receive action pops the message from this buffer.
An  execution $e=a_1\cdots a_n$ is a sequence of actions in
$\sendSet \cup \receiveSet$ such that $(\globalstate{l_0}, \B_0) \xrightarrow{a_1} \cdots \xrightarrow{a_n} (\globalstate{l}, \B)$ for some $\globalstate{l}$ and $\B$.
As usual $\xRightarrow{e}$ stands for $ \xrightarrow{a_1} \cdots \xrightarrow{a_n}$.
We write $\asEx(\system)$ to denote the set of asynchronous executions of
a system $\system$.
In a sequence of actions $e=a_1\cdots a_n$,
a send action $a_i= \send{p}{q}{\amessage}$ is  \emph{matched}
by a reception $a_j=\rec{p'}{q'}{\amessage'}$  (denoted by $a_i \matches a_j$)
if $i< j$, $p=p'$, $q=q'$, $\amessage=\amessage'$, and there is $\ell\geq 1$ such that
$a_i$ and $a_j$ are the $\ell$th actions of $e$ with these properties respectively.
A send action $a_i$ is \emph{unmatched} if there is no
matching reception in $e$.
A \emph{message exchange} of a sequence of actions
$e$ is a set either of the form  $v=\{a_i,a_j\}$ with $a_i\matches a_j$ or
of the form $v=\{a_i\}$ with $a_i$ unmatched. 
For a message $\amessage_i$, we will note $v_i$ the corresponding message exchange.
When $v$ is either an unmatched
$\send{p}{q}{\amessage}$ or a pair of matched actions
$\{\send{p}{q}{\amessage},\rec{p}{q}{\amessage}\}$, we write
$\sender{v}$ for $p$ and $\receiver{v}$ for $q$.
Note that $\procofactionv{R}{v}$ is defined even if $v$ is unmatched.
Finally, we write $\procs{v}$ for $\{p\}$ in the case of an unmatched send
and $\{p,q\}$ in the case of a matched send.

An execution imposes a total order on the actions. Here, we are interested in stressing the causal dependencies between messages. We thus make use of  message sequence charts (MSCs) that only impose an order
between matched pairs of actions and between the actions of a same
process. Informally, an MSC will be depicted with vertical
timelines (one for each process) where time goes from top to bottom, that carry some events (points) representing send and receive actions of this process (see Fig.~\ref{fig:exmp:causal}). An arc is
drawn between two matched events. We will also draw a dashed arc
to depict an  unmatched send event.
An MSC is, thus,  a partially ordered set of events,
each corresponding to a send or receive action.

\begin{definition}[MSC]\label{def:msc}
  A message sequence chart 
  is a tuple
  $(Ev, \lambda, \prec)$,
  where
\begin{itemize}
\item $Ev$ is a finite set of events,
\item $\lambda :Ev\to\sendSet\cup\receiveSet$ tags each event with an action,
\item $\prec=(\prec_{po}\cup \prec_{src})^+$ is the transitive closure of $\prec_{po}$ and $\prec_{src}$ where: 
\begin{itemize}
\item $\prec_{po}$ is a partial order on $Ev$ such that, for all process $p$,
$\prec_{po}$ induces a total order on the set of events of process $p$, \ie on $\lambda^{-1}(\sendSet_p\cup \receiveSet_p)$
\item $\prec_{src}$ is a binary relation that relates each receive event to its preceding send event :
\begin{itemize}
\item for all event $r\in \lambda^{-1}(\receiveSet)$, there is exactly one event $s$ such that $s\prec_{src} r$
\item for all event $s\in \lambda^{-1}(\sendSet)$, there is at most one event $r$ such that $s\prec_{src} r$
\item for any two events $s,r$ such that $s\prec_{src}r$, there are $p,q,\amessage$ such that $\lambda(s)=\send p q \amessage$ and $\lambda(r)=\rec p q \amessage$.
\end{itemize}
\end{itemize}
\end{itemize}
\end{definition}
We identify \MSC{}s up to graph isomorphism (i.e., we view an
\MSC as a labeled graph).
For a given \emph{well-formed}  (\ie  each reception is matched) sequence of actions $e=a_1\dots a_n$, we let
$msc(e)$ be the MSC where $Ev=\interval{1}{n}$, $\prec_{po}$ is the set of pairs of indices $(i,j)$ such
that $i<j$ and 
 $\{a_i,a_j\}\subseteq \sendSet_p\cup\receiveSet_p$ for some $p\in\procSet$ (\ie $a_i$ and $a_j$ are actions of a same process), and
$\prec_{src}$ is the set of pairs of indices $(i,j)$ such that
$a_i\matches a_j$. We say that $e=a_1\dots a_n$ is a \emph{linearisation} of $msc(e)$, and
we write $\asTr(\system)$ to denote $\{ msc(e) \mid e \in \asEx(\system) \}$  the set of MSCs of system $\system$.

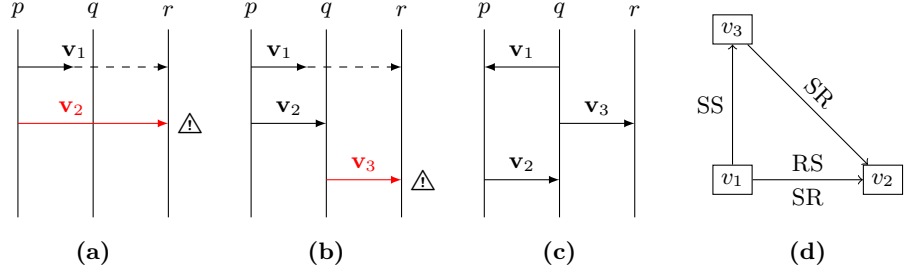
\begin{figure}[t]
\begin{tikzpicture}
\begin{scope}
	\coordinate (pa) at (0,-0.25) ;
	\coordinate (pb) at (0,-2.75) ;
	\coordinate (qa) at (1,-0.25) ;
	\coordinate (qb) at (1,-2.75) ;
	\coordinate (ra) at (2, -0.25);
	\coordinate (rb) at (2, -2.75);
		\draw (0,0) node{$p$} ;
	\draw (1,0) node{$q$} ;
	\draw (2,0) node{$r$} ;
	\draw (pa) -- (pb) ;
	\draw (qa) -- (qb) ;
	\draw (ra) -- (rb) ;

	\coordinate (as) at (0,-0.75);
	\coordinate (ar) at (0.75,-0.75);
	\draw[>=latex,->] (as) -- (ar) node[above, sloped] {$\amessage_1$};
	\draw[>=latex,->,  dashed] (ar) -- (2, -0.75);

	\coordinate (ds) at (0, -1.5);
	\coordinate (dr) at (2, -1.5);
        \node at (dr) [right] {\danger};
	\draw[>=latex,->, red] (ds) to node[above left,red]{$\amessage_2$} (dr);

	\draw(1, -3.5) node[above]{\textbf{(a)}};
\end{scope}
\begin{scope}[shift = {(3.1,0)}]
	\coordinate (pa) at (0,-0.25) ;
	\coordinate (pb) at (0,-2.75) ;
	\coordinate (qa) at (1,-0.25) ;
	\coordinate (qb) at (1,-2.75) ;
	\coordinate (ra) at (2, -0.25);
	\coordinate (rb) at (2, -2.75);
		\draw (0,0) node{$p$} ;
	\draw (1,0) node{$q$} ;
	\draw (2,0) node{$r$} ;
	\draw (pa) -- (pb) ;
	\draw (qa) -- (qb) ;
	\draw (ra) -- (rb) ;

	\coordinate (as) at (0,-0.75);
	\coordinate (ar) at (0.75,-0.75);
	\draw[>=latex,->] (as) -- (ar) node[midway, above] {$\amessage_1$};
	\draw[>=latex,->,  dashed] (ar) -- (2, -0.75);

	\coordinate (bs) at (0, -1.5);
	\coordinate (br) at (1, -1.5);
	\draw[>=latex,->] (bs) -- (br) node[midway, above] {$\amessage_2$};

	\coordinate (cs) at (1,-2.25);
	\coordinate (cr) at (2,-2.25);
        \node at (cr) [right] {\danger};
	\draw[>=latex,->, red] (cs) -- (cr) node[midway, above] {$\amessage_3$};
		\draw(1, -3.5) node[above]{\textbf{(b)}};

\end{scope}
\begin{scope}[shift = {(6.2,0)}]
    \coordinate(pa) at (0,-0.25) ;
    \coordinate (pb) at (0,-2.75) ;
    \coordinate (qa) at (1,-0.25) ;
    \coordinate (qb) at (1,-2.75) ;
    \coordinate (ra) at (2,-0.25) ;
    \coordinate (rb) at (2,-2.75) ;

			\draw (0,0) node{$p$} ;
	\draw (1,0) node{$q$} ;
	\draw (2,0) node{$r$} ;
    \draw (pa) -- (pb) ;
    \draw (qa) -- (qb) ;
    \draw (ra) -- (rb) ;

    \coordinate (s1) at (1,-0.75);
    \coordinate (r1) at (0,-0.75);
    \draw[>=latex,->] (s1) -- node [above,sloped] {$\amessage_1$} (r1);
    \coordinate (s2) at (0, -2.25);
    \coordinate (r2) at (1, -2.25);
    \coordinate (s3) at (1, -1.5);

    \coordinate (r3) at (2, -1.5);

    \draw[>=latex,->] (s2) -- node [above,sloped] {$\amessage_2$} (r2);
    \draw[>=latex,->] (s3) -- node [above,sloped] {$\amessage_3$} (r3);
    \draw(1, -3.5) node[above]{\textbf{(c)}};

\end{scope}
\begin{scope}[shift = {(9.5, -2.25)}]
	\node[draw] (m1) at (0,0) {$v_1$};
	\node[draw] (m2) at (2,0) {$v_2$};
	\node[draw] (m3) at (0,2) {$v_3$};

	\draw[->] (m1) -- node [above] {RS} node [below] {SR} (m2);
	\draw[->] (m1) -- node [left] {SS} (m3);
	\draw[->] (m3) -- node [above,sloped] {SR} (m2);
		\draw(1, -1.25) node[above]{\textbf{(d)}};

\end{scope}
\end{tikzpicture}
\vspace*{-.2cm}
\caption{(a) and (b): two MSCs that violate causal delivery. 
  (c) and (d): an MSC and its conflict graph 
\vspace*{-0.5cm}
}
\label{fig:exmp:causal}
\end{figure}

Mailbox communication imposes a number of constraints on what and when messages can be read.  The precise definition is given below, we now discuss some of the possible scenarios.  For instance:
  if two messages are sent to a same process,  they will be received in the same order as they have been sent. As another example, unmatched messages also
  impose some constraints:
  if a process $p$ sends an unmatched message to $r$, it will not  be able to send matched messages to $r$ afterwards (Fig.~1a);
 or similarly,
    if a process $p$ sends an unmatched message to $r$, any process $q$ that receives subsequent messages from $p$ will not  be able to send matched messages to $r$ afterwards (Fig.~1b).
When 
an MSC satisfies the constraint imposed by mailbox communication, we say that it satisfies causal delivery. Notice that, by construction, all executions satisfy causal delivery. 

\begin{definition}[Causal delivery]\label{def:causal-delivery}
Let $(Ev , \lambda, \prec)$ be an MSC. 
  We say that it satisfies causal delivery if
  the MSC has a linearisation $e=a_1\dots a_n$ such that
  for any two  events $i\prec j$ such that $a_i~=~\send{p}{q}{\amessage}$ and $a_j = \send{p'}{q}{\amessage'}$,  either $a_j$ is unmatched, or there are $i',j'$ such that $a_i\matches a_{i'}$, $a_j\matches a_{j'}$, and $i'\prec j'$.
\end{definition}

 Our definition enforces the following intuitive property. 
 
\begin{proposition}\label{prop:cau}
An MSC $msc$ satisfies causal delivery if and only if there is a system $\system$ and an execution $e\in\asEx(\system)$
such that $msc=msc(e)$.
\end{proposition}

We now recall from \cite{DBLP:conf/cav/BouajjaniEJQ18} the definition of \emph{conflict graph} depicting the causal dependencies between message exchanges. Intuitively, we have  a dependency whenever two messages have a  process in common. For instance an $\cgedge{SS}$ dependency between message exchanges $v$ and $v'$ expresses the fact that $v'$ has been sent  after $v$, by the same process.

%



\begin{definition}[Conflict graph]\label{def:conflict-graph}
  The conflict graph $\cgraph{e}$ of a sequence of actions $e=a_1\cdots a_n$ is the labeled
  graph $(V,\{\cgedge{XY}\}_{X,Y\in\{R,S\}})$ where $V$ is the set
  of message exchanges of $e$, and for all $X,Y\in\{S,R\}$, for all $v,v'\in V$,
  there is a $XY$ dependency edge $v\cgedge{XY}v'$ between $v$ and $v'$
  if there are $i<j$ such that $\{a_i\}= v\cap X$,
  $\{a_j\}=v'\cap Y$, and $\procofactionv{X}{v}=\procofactionv{Y}{v'}$.
\end{definition}

Notice that each linearisation $e$ of an MSC will have the same conflict graph. We can thus talk about an MSC and the associated conflict graph. (As an exam\-ple cfr. Figs.~1c and 1d.)

We write $v\to v'$ if $v\cgedge{XY} v'$ for some
$X,Y\in\{R,S\}$,
and $v\to^*v'$ if there is a (possibly empty) path from $v$  to $v'$.



\section{\kSable{k} systems}
\label{sec:ksity}

In this section, we define  \kSable{k} systems. 
The main contribution of this part is a new characterisation
of \kSable{k} executions  that corrects the one given in~\cite{DBLP:conf/cav/BouajjaniEJQ18}.

In the rest of the paper, $k$ denotes a given integer $k\geq 1$.
A \kE{k} denotes a sequence of actions  starting with
at most $k$ sends and  followed by at most $k$ receives matching some of
the sends.  An  MSC is \emph{\kSous{k}} if there exists a linearisation that is breakable into 
a sequence of \emph{\kE{k}s}, such that
a message sent during a \kE{k} cannot
be received during a subsequent one: either it is received
during the same \kE{k}, or it remains orphan forever.


\begin{definition}[\kSous{k}]
  An MSC $msc$ is \kSous{k} if:
  \begin{enumerate}
  \item there exists a linearisation of $msc$ $e=e_1\cdot e_2\cdots e_n$ where for all $i\in\interval{1}{n}$, $e_i\in\sendSet^{\leq k}\cdot \receiveSet^{\leq k}$,
  \item $msc$ satisfies causal delivery,
  \item for all $j,j'$ such that $a_j\matches a_{j'}$ holds in $e$,
    $a_j\matches a_{j'}$ holds in some $e_i$.
  \end{enumerate}
  An execution $e$ is \kSable{k} if $msc(e)$ is \kSous{k}.
\end{definition}

We write $\sTr_k(\system)$ to denote the set 
$\{msc(e)\mid e\in\asEx(\system)\mbox{ and } msc(e)$  is \kSous{k}\}.



\begin{example}[\kSous{k} MSCs and \kSable{k} executions]\ \label{exmp:ksynchronous}
\begin{enumerate}
\item \label{exmp:nonksable} There is no $k$ such that the  MSC  in Fig.~\ref{fig:ex-counter-exmp-SCC}a is  \kSous{k}. All messages must be grouped in the same \kE{k}, but it is not
  possible to schedule all the sends first, because the reception of
  $\amessage_1$ happens before the sending of $\amessage_3$. Still, this MSC
  satisfies causal delivery.
\item \label{exmp:ksable}

Let $e_1 = \send{r}{q}{\amessage_3}\cdot \send{q}{p}{\amessage_2} \cdot \send{p}{q}{\amessage_1}\cdot\rec{q}{p}{\amessage_2}\cdot\rec{r}{q}{\amessage_3}$ be an execution. Its MSC, $msc(e_1)$  depicted in Fig.~\ref{fig:ex-counter-exmp-SCC}b satisfies causal delivery.
Notice that $e_1$ can not be divided in \kE{1}s. However, if we consider the alternative linearisation of $msc(e_1)$:  $e_2 = \send{p}{q}{\amessage_1}\cdot\send{q}{p}{\amessage_2}\cdot\rec{q}{p}{\amessage_2}\cdot\send{r}{q}{\amessage_3}\cdot\rec{r}{q}{\amessage_3}$,  we have that  $e_2$ is breakable into \kE{1}s in which each matched send is in a \kE{1} with its reception. Therefore, $msc(e_1)$ is \kSous{1} and $e_1$ is \kSable{1}. 
Remark that $e_2$ is not  an execution and there exists no execution that can be divided into \kE{1}s. A  \kSous{k}  MSC  highlights dependencies between messages but does not impose an order for the execution.



   \end{enumerate}
\end{example}

\begin{comp}
In \cite{DBLP:conf/cav/BouajjaniEJQ18}, the authors define the set $sEx_{k}(\system)$ for a system $\system$ as the set of \kSous{k} executions of the system in the \kSous{k} semantics. 
Nonetheless as remarked in Example \ref{exmp:ksynchronous}.2 not all  executions of a system can be divided into \kE{k}s even if they are \kSable{k}.
Thus, in order not to lose any executions, we have decided to reason only on MSCs (called traces in \cite{DBLP:conf/cav/BouajjaniEJQ18}).
%
\end{comp}

Following standard terminology, we say that
a set $U\subseteq V$ of vertices is a \emph{strongly
  connected component} (SCC) of a given graph $(V,\to)$
if between any two vertices $v,v'\in U$,
there exist
two oriented paths $v\to^*v'$ and $v'\to^* v$.
The statement below fixes some issues with
Theorem~1 in~\cite{DBLP:conf/cav/BouajjaniEJQ18}.

\begin{theorem}[Graph-theoretic characterisation of \kSous{k} MSCs]
  \label{thm:k-sity-scc}
  Let $msc$  be a causal delivery MSC.
   $msc$ is \kSous{k} iff every SCC in its
  conflict graph is of size at most $k$ and if  no RS edge occurs on any
  cyclic path.
\end{theorem}
\iflong
\begin{proof}
\begin{quote}
{\bf Theorem 1}
   Let $msc$  be a causal delivery MSC.
   $msc$ is \kSous{k} iff every SCC in its
  conflict graph is of size at most $k$ and if  no RS edge occurs on any
  cyclic path.
\end{quote}

\begin{proof}
  Let $msc$ be a causal delivery MSC. 
  $\implies$ If $msc$ is \kSous{k},
  then $\exists e = e_1 \cdots e_n$ such that $msc(e) = msc$ where  
  each $e_i$ is a \kE{k}. For every vertex $v$ of the conflict graph 
  $\cgraph{e}$ 
  there is exactly one index $\iota(v)\in [1..n]$ such that
  $v\subseteq e_{\iota(v)}$.
  Now, observe that if there is an edge from $v$ to $v'$ in the conflict graph,
  some action of $v$ must happen before some action of $v'$,
  i.e., $\iota(v)\leq \iota(v')$. So if $v,v'$ are on a same SCC, $\iota(v)=\iota(v')$,  they must both
  occur within the same \kE{k}. Since each \kE{k} contains at most
  $k$ message exchanges, this shows that all SCC are of size at most $k$.
  Observe also that if $v\cgedge{RS}v'$, then $\iota(v)<\iota(v')$, since
  within a \kE{k} all the sends precede all the receives.
  So an RS edge cannot occur on a cyclic path.

  $\Longleftarrow$ 
  Let $e$ be a linearisation of $msc$. 
  Assume now that conflict graph $\cgraph{e}$ 
  neither contains a SCC of size greater than $ k$ nor a cyclic path
  with an RS edge.
  Let $V_1,\dots,V_n$ be the set of maximal SCCs of the conflict graph,
  listed in some topological order. For a fixed $i$,
  let
  $e_i=s_1\dots s_mr_1\dots r_{m'}$ be the enumeration of the actions
  of the message exchanges of $V_i$ defined by taking
  first all send actions of $V_i$ in the order in which they appear in $e$,
  and second all the receive actions of $V_i$ in the same order as
  in $e$. Let $e'=e_1\dots e_n$. Then $\cgraph{e'}$ is the same as 
  $\cgraph{e}$: indeed, the permutation of
  actions we defined could only postpone a receive after a send of
  a same SCC,
  therefore it could only replace some $v\cgedge{RS}v'$ edge
  with an $v'\cgedge{SR}v$ edge between two vertices $v,v'$
  of a same SCC, but we assumed that the SCCs do not contain RS edges,
  so it does not happen.
  Therefore $e$ and $e'$ have the same conflict graph,
  and $msc(e')=msc(e)$. Moreover, also by hypothesis,
  $|V_i|\leq k$ for all $i$, therefore each $e_i$ is a \kE{k},
  and finally 
  $msc$ is \kSous{k}. \qed
\end{proof}
\end{proof}
\fi

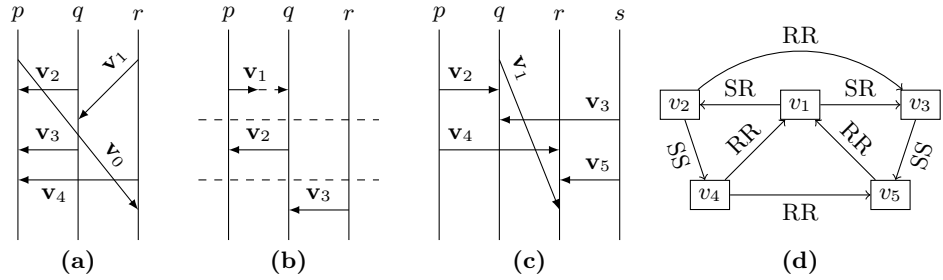
\begin{figure}[t]
\begin{tikzpicture}[scale = 0.8]
	\begin{scope}
		\draw (0,0) to (0,-3.5) ;
		\draw (1,0) to (1,-3.5) ;
		\draw (2,0) to (2,-3.5) ;
		\draw (0,0.25) node{$p$} ;
		\draw (1,0.25) node{$q$} ;
		\draw (2,0.25) node{$r$} ;

		\draw[->,>=latex] (0,-.5) to node [above, pos=0.7,sloped] {$\amessage_0$} (2,-3);
		\draw[->,>=latex] (2,-.5) to node [above,sloped, pos = 0.2] {$\amessage_1$} (1,-1.5);
		\draw[->,>=latex] (1,-1) to node [above] {$\amessage_2$} (0,-1);
		\draw[->,>=latex] (1,-2) to node [above] {$\amessage_3$} (0,-2);
		\draw[->,>=latex] (2,-2.5) to node [below, pos=.7] {$\amessage_4$} (0,-2.5);
  		\draw (1,-4.2) node[above]{\textbf{(a)}};
	\end{scope}
	\begin{scope}[shift={(3.5,0)}]
		\draw (0,0) to (0,-3.5) ;
		\draw (1,0) to (1,-3.5) ;
		\draw (2,0) to (2,-3.5) ;
		\draw (0,0.25) node{$p$} ;
		\draw (1,0.25) node{$q$} ;
		\draw (2,0.25) node{$r$} ;
		\draw[->,>=latex] (0,-1) to node [above, pos = 0.9] {$\amessage_1$} (.5,-1);
		\draw[->,>=latex,dashed] (.5,-1) to (1,-1);
		\draw[->,>=latex] (1,-2) to node [above] {$\amessage_2$} (0,-2);
		\draw[->,>=latex] (2,-3) to node [above] {$\amessage_3$} (1,-3);
   		\draw (1,-4.2) node[above]{\textbf{(b)}};
                \draw[dashed] (-0.5,-1.5) -- (2.5,-1.5);
                \draw[dashed] (-0.5,-2.5) -- (2.5,-2.5);
	\end{scope}
	\begin{scope}[shift = {(7,0)}]
		\coordinate(pa) at (0,0) ;
		\coordinate (pb) at (0,-3.5) ;
		\coordinate (qa) at (1,0) ;
		\coordinate (qb) at (1,-3.5) ;
		\coordinate (ra) at (2,0) ;
		\coordinate (rb) at (2,-3.5) ;
		\coordinate (sa) at (3,0);
		\coordinate (sb) at (3, -3.5);

		\draw (0,0.25) node{$p$} ;
		\draw (1,0.25) node{$q$} ;
		\draw (2,0.25) node{$r$} ;
		\draw (3, 0.25) node{$s$} ;
		\draw (pa) -- (pb) ;
		\draw (qa) -- (qb) ;
		\draw (ra) -- (rb) ;
		\draw (sa) -- (sb) ;

		\coordinate (s1) at (1,-0.5);
		\coordinate (r1) at (2,-3);
		\draw[>=latex,->] (s1) to node [above,sloped, pos = 0.1]{$\amessage_1$} (r1);

		\coordinate (s2) at (0, -1);
		\coordinate (r2) at (1, -1);
		\draw[>=latex,->] (s2) -- (r2);
		\draw (s2) node[above right]{$\amessage_2$};

		\coordinate (s3) at (3, -1.5);
		\coordinate (r3) at (1, -1.5);
		\draw[>=latex,->] (s3) -- (r3);
		\draw (s3) node[above left] {$\amessage_3$};

		\coordinate (s4) at (0, -2);
		\coordinate (r4) at (2, -2);
		\draw[>=latex,->] (s4) -- (r4);
		\draw (s4) node[above right]{$\amessage_4$};

		\coordinate (s5) at (3, -2.5);
		\coordinate (r5) at (2, -2.5);
		\draw[>=latex,->] (s5) -- (r5);
		\draw (s5) node[above left]{$\amessage_5$};

		\draw (1.5,-4.2) node[above]{\textbf{(c)}};
 	\end{scope}
 	\draw (13,-4.2) node[above]{\textbf{(d)}};
 	\begin{scope}[shift={(13,-1.25)}]
   		\node[draw] (a) at (0,0) {$v_1$};
		\node[draw] (b) at (-2,0) {$v_2$};
		\node[draw] (c) at (2,0) {$v_3$};
		\node[draw] (d) at (-1.5,-1.5) {$v_4$};
		\node[draw] (e) at (1.5,-1.5) {$v_5$};

		\draw[->] (a) to node[sloped, midway, above]{SR} (b);
		\draw[->] (a) to node[sloped, midway, above]{SR} (c);
		\draw[->] (b) to node[sloped, midway, below] {SS} (d);
		\draw[->] (d) to node[sloped, midway, below]{RR} (e);
		\draw[->] (d) to node[sloped, midway, above]{RR} (a);
		\draw[->] (e) to node[sloped, midway, above]{RR} (a);
		\draw[->] (c) to node[sloped, midway, below]{SS} (e);
		\draw[->] (b) to[bend left = 40] node[sloped, midway, above]{RR} (c);
	\end{scope}
	\end{tikzpicture}
\caption{(a) the MSC of Example~\ref{exmp:ksynchronous}.\ref{exmp:nonksable}. (b) the MSC of Example~\ref{exmp:ksynchronous}.\ref{exmp:ksable}. (c) the MSC of Example~\ref{ex:counter-exmp-SCC} and (d) its conflict graph.}
\label{fig:ex-counter-exmp-SCC}
\end{figure}

\begin{example}[A \kSous{5} MSC]\label{ex:counter-exmp-SCC}
  Fig.~\ref{fig:ex-counter-exmp-SCC}c depicts a \kSous{5} MSC,  that is not \kSous{4}. Indeed, its conflict graph (Fig.~\ref{fig:ex-counter-exmp-SCC}d) contains a
  SCC of size 5 (all vertices are on the same SCC).
\end{example}

\begin{comp}\label{rmk:characterisation} 
Bouajjani \textit{et al.}
give a characterisation of \kSous{k} executions similar to ours, but they use the word \emph{cycle} instead of SCC, and the subsequent developments of the paper suggest that they intended to say
\emph{Hamiltonian cycle} (i.e., a cyclic path that does not go twice
through the same vertex). It is not the case that a MSC is \kSous{k} if and only if
every Hamiltonian cycle in its
conflict graph is of size at most $k$ and if  no RS edge occurs on any
cyclic path. Indeed,
consider again Example~\ref{ex:counter-exmp-SCC}. This graph is not Hamiltonian,
and the largest Hamiltonian cycle indeed is of size 4 only. But as we already
discussed in Example~\ref{ex:counter-exmp-SCC}, the corresponding MSC is not
\kSous{4}.

As a consequence, the algorithm that is presented in \cite{DBLP:conf/cav/BouajjaniEJQ18}
for deciding whether a system is \kSable{k} is not correct as well: the MSC of Fig.~\ref{fig:ex-counter-exmp-SCC}c would
be considered  \kSous{4} according to this algorithm, but it is not.
\end{comp}


\section{Decidability of reachability for \kSable{k} systems}
\label{sec:reachability} 

We show  that the reachability problem is decidable for \kSable{k} systems.
While proving this result, we have to face several non-trivial aspects of
causal delivery that were missed in~\cite{DBLP:conf/cav/BouajjaniEJQ18}
and that require a completely new approach. 


\begin{definition}[\kSable{k} system]
A system $\system$ is \emph{\kSable{k}} if all its executions are \kSable{k}, 
i.e., $\sTr_k(\system) = \asTr(\system)$. 
\end{definition}

In other words, a system $\system$ is \kSable{k} if for every
execution $e$ of $\system$, $msc(e)$ may be divided into \kE{k}s. 
\begin{remark}
In particular, a system may be \kSable{k} even if
some of its executions fill the buffers with more than $k$ messages.
For instance, the only linearisation of the \kSous{1} MSC Fig.~\ref{fig:ex-counter-exmp-SCC}b that is an execution of the system needs buffers of size 2. 
\end{remark}

 For a \kSable{k} system, the reachability problem reduces to the rea\-cha\-bility through a  \kSable{k} execution. To show that
\kSous{k} reachability is decidable, we establish that
the set of \kSous{k} MSCs is regular.
More precisely, we want to define a finite state automaton that
accepts a sequence
$e_1\cdot e_2\cdots e_n$ of \kE{k}s if and only if they satisfy causal delivery.

We start by giving a graph-theoretic characterisation of causal delivery.
For this, we define the \emph{extended edges}
$v\cgedgeD{XY}v'$ of a given conflict graph.
The relation $\cgedgeD{XY}$ is
defined in Fig.~\ref{fig:conflict-graph-extension} with $X,Y \in \{S,R\}$. Intuitively,
$v \cgedgeD{XY}v'$ expresses that 
 event  $X$ of $v$
must happen before  event  $Y$ of $v'$ due to either their
order on the same machine (Rule 1), or the fact that a send
happens before its matching receive (Rule 2), or due to the mailbox
semantics (Rules 3 and 4), or because of a chain of such dependencies (Rule 5).
We observe that in the \emph{extended conflict graph}, obtained applying such rules, a cyclic dependency appears whenever causal delivery is not satisfied.

\begin{figure}[t]
\begin{center}
\begin{minipage}{.33\columnwidth}
\begin{center}
\AxiomC{\small $v_1 \cgedge{XY} v_2$}  
\LeftLabel{\small (Rule 1)}
\UnaryInfC{\small $v_1 \cgedgeD{XY} v_2$}
\DisplayProof
\end{center}
\end{minipage}
\begin{minipage}{.33\columnwidth}
\begin{center}
\AxiomC{\small$v \cap  R \neq \emptyset$}
\LeftLabel{\small(Rule 2)}
\UnaryInfC{\small$v \cgedgeD{SR} v$}
\DisplayProof
\end{center}
\end{minipage}
\begin{minipage}{.33\columnwidth}
\begin{center}
\AxiomC{\small$v_1 \cgedge{RR} v_2$}
\LeftLabel{\small(Rule 3)}
\UnaryInfC{\small$v_1 \cgedgeD{SS} v_2$}
\DisplayProof
\end{center}
\end{minipage}
\end{center}

\begin{center}
\begin{minipage}{.5\columnwidth}
\begin{center}
 \AxiomC{\small
    \begin{tabular}{c}
	$v_1 \cap R \neq \emptyset$
    \qquad
    $v_2 \cap R = \emptyset$
    \\
    $\receiver{v_1}=\receiver{v_2}$
  \end{tabular}}
  \LeftLabel{\small(Rule 4)}
\UnaryInfC{\small$v_1 \cgedgeD{SS} v_2$}
\DisplayProof
\end{center}
\end{minipage}
\begin{minipage}{.5\columnwidth}
\begin{center}
\AxiomC{\small$v_1 \cgedgeD{XY}\cgedgeD{YZ} v_2$}
\LeftLabel{\small(Rule 5)}
\UnaryInfC{\small$v_1 \cgedgeD{XZ} v_2$}
\DisplayProof
\end{center}
\end{minipage}
\end{center}

\vspace*{-0.5cm}
\caption{\label{fig:rules-cgedgeD}Deduction rules for extended dependency edges of the conflict graph \vspace*{-0.5cm}}
\label{fig:conflict-graph-extension}
\end{figure}


 \begin{example}\label{ex:causal-delivery} 
 Fig.~\ref{fig:exmp-invalid-execution-B}a and \ref{fig:exmp-invalid-execution-B}b 
  depict  an MSC and its associated conflict graph with
  some extended edges. 
  This MSC violates causal delivery and  there is a cyclic dependency
  $v_1\cgedgeD{SS}v_1$. \end{example}

\begin{theorem}[Graph-theoretic characterisation of causal delivery]\label{thm:causal-delivery-graphically}
  An MSC 
  satisfies causal delivery iff
  there is no cyclic causal dependency of the form $v\cgedgeD{SS}v$ for
  some vertex $v$ of its extended conflict graph.
\end{theorem}
\iflong
\begin{proof}
\begin{quote}
{\bf Theorem 2 }
  An MSC   
  satisfies causal delivery iff
  there is no cyclic causal dependency of the form $v\cgedgeD{SS}v$ for
  some vertex $v$ of the associated extended conflict graph.
\end{quote} 
 
 \begin{proof}
  $\Rightarrow$ 
  Assume that $msc$ satisfies causal delivery. Then there is
  a total order $\totord$ on the events that is a linearisation of
  $\prec=(\prec_{po}\cup \prec_{src})^+$ (cfr. Definition~\ref{def:msc}) with the property
  stated in Definition~\ref{def:causal-delivery}.
  We claim that if $v\cgedgeD{XY}v'$, $\{a_i\}=v\cap X$ and $\{a_j\}=v'\cap Y$,
  then $i\totord j$. The proof of this claim is by 
  induction on the derivation tree of $v\cgedgeD{XY}v'$:
  \begin{itemize}
  \item case of Rule 1 : $(i,j)\in \prec_{po}$, so $i\totord j$;
  \item case of Rule 2 : $(i,j)\in \prec_{src}$, so $i\totord j$;
  \item cases of Rules 3 and 4 : by definition of causal delivery;
  \item case of Rule 5 : there is $v_3$ such that $v_1\cgedgeD{XZ}v_3\cgedge{ZY}v_2$. Let $a_l$ be the $Z$ action of $v_3$. By inductive hypothesis,
    $i\totord{l}\totord{j}$, and by transitivity of $\totord$, $i\totord j$.
  \end{itemize}
  So we proved our claim, and $\totord$ extends $\cgedgeD{XY}$.
  As a consequence, there is no $\cgedgeD{SS}$ cycle.
  \\
  $\Leftarrow$ Assume that the extended dependency graph does not contain
  any $\cgedgeD{SS}$ cycle. Let us first show that it does not
  contain any $\cgedgeD{RR}$ cycle either. By contradiction assume there
  is some $v$ such that $v\cgedgeD{RR}v$. Since there is no $\cgedgeD{SS}$ cycle,
  there is no $v'$ on the cyclic path such that
  $v\cgedgeD{RS}v'\cgedgeD{SR}v$. So $v(\cgedge{RR})^*v$, and we reach a contradiction,
  as $\cgedge{RR}$ is included in $\prec_{po}$ which is acyclic.
  So $\cgedgeD{RR}$ is acyclic, and $\cgedgeD{XY}$ defines a
  partial order on actions. 
  Let us pick some linearisation of that order, and let $\totord$
  denote the associated order on indexes, i.e., $\totord$ is a total
  order such that for any $X$ action $a_i\in v_i$ and $Y$ action $a_j\in v_j$,
  $v_i\cgedgeD{XY}v_j$ implies $i\totord j$. We want to show that
  $\totord$ satisfies the property of Definition~\ref{def:causal-delivery}.
  Let $i\totord j$ with $a_i,a_j\in S$ and $\receiver{a_i}=\receiver{a_j}$,
  and let $v_i,v_j$ be the two vertices such that $a_i\in v_i$ and $a_j\in v_j$.
  Since $\totord$ extends $\cgedgeD{XY}$, either
  $v_i\cgedgeD{SS} v_j$ or $\neg(v_j\cgedgeD{SS}v_i)$.
  \begin{itemize}
  \item
    Assume  that $v_i\cgedgeD{SS}v_j$. If $v_i$ is unmatched, then $v_j$ must be unmatched otherwise by Rule 4 we would have $v_j\cgedgeD{SS} v_i$, which would
    violate the acyclicity hypothesis.
    On the other hand, if both $v_i$ and $v_j$ are matched, then $v_i\cgedge{RR} v_j$, otherwise we would have  $v_j\cgedge{RR} v_j$ and by Rule 3 $v_j\cgedgeD{SS} v_j$, which would
    violate the acyclicity hypothesis. So there are $i',j'$ such that $v_i=\{a_i,a_{i'}\}$, $v_j=\{a_j,a_{j'}\}$ and $i'\totord j'$, as required by
    Definition~\ref{def:causal-delivery}.
\item
  Assume  that $\neg(v_i\cgedgeD{SS}v_j)$
  and $\neg(v_j\cgedgeD{SS}v_i)$. 
  Then both sends are unmatched (because of Rules 3 and 4),
  therefore the property
  of Definition~\ref{def:causal-delivery} holds, concluding the proof.
  \end{itemize}\qed
\end{proof}

\end{proof}
\fi

Let us now come back to our initial problem: we want to recognise with  finite memory the
sequences $e_1, e_2 \dots e_n$ of \kE{k}s that composed give an MSC that satisfies causal delivery.
 We proceed by reading each \kE{k} one by one in sequence.
 This entails that, at each step, we have only a partial view of the global 
 conflict graph. Still, we want to determine whether the
acyclicity condition of
Theorem~\ref{thm:causal-delivery-graphically} is satisfied in the global conflict graph.
The crucial observation is
that  only the edges  generated by Rule 4  may ``go back in time''.
This means that we have to remember enough information 
from the previously examined \kE{k}s to determine whether the current \kE{k}
 contains a vertex $v$ that shares an edge with
 some unmatched vertex $v'$ seen in a previous \kE{k} and whether 
 this could participate in a cycle.
This is achieved by 
computing two sets of processes $C_{S,p}$ and $C_{R,p}$ that collect the following information:
  a process $q$ is in $C_{S,p}$ if it performs a send action
causally after an unmatched send to $p$, or it is the sender of the unmatched
send;
a process $q$ belongs to $C_{R,p}$ if
it receives a message that was sent after some unmatched message
directed to $p$.
More precisely, 
we  have:
\begin{align*}
C_{S,p} = &\  \{\procofactionv{S}{v} \mid v' \cgedgeD{SS} v \mbox{ \& } v' 
\mbox{ is unmatched \& } \procofactionv{R}{v'} = p  \} \\
C_{R,p} =  &\  \{\procofactionv{R}{v} \mid v' \cgedgeD{SS} v \mbox{ \& } v' \mbox{ is unmatched \& } \procofactionv{R}{v'} = p \mbox{ \& }  v \cap R \neq \emptyset \}
\end{align*}

These sets  abstract and carry from one \kE{k} to another the necessary information to detect violations of causal delivery. We want
to compute them in any local conflict graph of a \kE{k} incrementally, i.e., 
knowing what they were at the end of the previous \kE{k}, we want to
compute them at the end of the current one. More precisely, let $e=s_1\cdots s_m\cdot r_1\cdots r_{m'}$ be a  \kE{k},   $\cgraph{e}=(V,E)$  its conflict graph
 and  
 $B:\procSet\to(2^{\procSet}\times 2^{\procSet})$ the function that 
associates to each $p\in\procSet$ the two sets
$B(p)=(C_{S,p},C_{R,p})$. Then, the conflict graph $\cgraph{e,B}$ is the graph
$(V',E')$ with $V'=V\cup\{\lambdanode_p\mid p\in\procSet\}$ and $E'\supseteq E$ as 
defined below.
For each
process $p\in \procSet$,
the ``summary  node'' $\lambdanode_p$  shall account for
all past unmatched messages sent to $p$ that occurred in some \kE{k} before $e$.
$E'$ is the set $E$ of edges
$\cgedge{XY}$ among message exchanges of $e$, as in Definition~\ref{def:conflict-graph},
augmented with the following set of extra edges taking into
account the summary nodes.
\begin{align}
  & \{\lambdanode_p \cgedge{SX} v \mid\procofactionv{X}{v} \in C_{S,p} 
  \mbox{ \& } v\cap X\neq\emptyset \mbox{ for some }X\in\{S,R\}\}  \label{eq:SS}
  \\ \cup \  &
 \{\lambdanode_p \cgedge{SS} v\mid \procofactionv{X}{v} \in C_{R,p} 
 \mbox { \& } v \cap R \neq \emptyset \mbox{ for some }X\in\{S,R\}\}  \label{eq:RR}
  \\ \cup \ & 
\{ \lambdanode_p \cgedge{SS} v \mid  \procofactionv{R}{v} \in C_{R,p}  \mbox{ \& }  v \mbox{ is unmatched} \}   \label{eq:unmatchedSS}
  \\ \cup \ &
 \{ v \cgedge{SS} \lambdanode_p \mid \procofactionv{R}{v} = p \mbox{ \& }   v \cap R \neq \emptyset   \} 
\ \cup \
\{ \lambdanode_q \cgedge{SS} \lambdanode_p \mid p \in C_{R,q}  \}  \label{eq:psiSS}
\end{align}
These extra edges summarise/abstract the connections to and from previous \kE{k}s.
Equation (\ref{eq:SS}) considers connections   $\cgedge{SS}$ and $\cgedge{SR}$ that are due to two sends messages or, respectively, a send and a receive on the same process. 
Equations (\ref{eq:RR}) and (\ref{eq:unmatchedSS}) considers connections   $\cgedge{RR}$ and $\cgedge{RS}$ that are due to two received messages or, respectively, a receive and a subsequent send on the same process. Notice how the rules in Fig. \ref{fig:conflict-graph-extension} would then imply the existence of a connection $\cgedgeD{SS}$, in particular Equation (\ref{eq:unmatchedSS}) abstract the existence of an edge $\cgedgeD{SS}$ built because of Rule 4.
Equations in (\ref{eq:psiSS})
abstract edges that would connect the current \kE{k} to previous ones. 
As before those edges in the global conflict graph would correspond to extended edges added because of Rule 4 in Fig.~\ref{fig:conflict-graph-extension}. 
Once we have this enriched local view of the conflict graph, we take its extended version. 
Let  $\cgedgeD{XY}$ denote  the edges of the extended conflict graph as defined from rules in Fig.~\ref{fig:conflict-graph-extension} taking
into account the new vertices $\lambdanode_p$ and their edges.

\begin{figure}[t]
\begin{prooftree}
\AxiomC{\small \stackanchor{ 
$e=s_1\cdots s_m\cdot r_1\cdots r_{m'}\quad s_1\cdots s_m\in S^* \quad r_1\cdots r_{m'}\in R^* \quad  0\leq m'\leq m \leq k$}
		{ \stackanchor	{$(\globalstate{l},\B_0)\xRightarrow{e}(\globalstate{l'},\B) \mbox{ for some } \B$}
			{ \stackanchor	{$\mbox{ for all }p\in\procSet\quad B(p)=(C_{S,p},C_{R,p}) \mbox{ and } B'(p)=(C_{S,p}',C_{R,p}'), $} 
				{\stackanchor { $\mathsf{Unm}_p=\{\lambdanode_p\}\cup\{v\mid v  \mbox { is unmatched,  } \receiver{v}=p\}$}
					{\stackanchor { $C_{X,p}'=C_{X,p}
					 \cup  \{p \mid p  \in C_{X,q}, v \cgedgeD{SS} \lambdanode_q, (\procofaction{R}{v} = p \mbox{ or } v = \lambdanode_p)\} ~\cup$} 
		{\stackanchor 	{
		\stackanchor{$ \{\procofaction{X}{v} \mid v \in\mathsf{Unm}_p\cap V, X = S \}\cup  \{\procofactionv{X}{v'} \mid v\cgedgeD{SS} v', v\in\mathsf{Unm}_p, v \cap X \neq \emptyset\}
		$}}
 		{$\mbox{for all } p \in \procSet, p\not\in C_{R,p}'$  }	
}}}}}}
\UnaryInfC{{\small$ (\globalstate l,B)\transitionKE{e}{k}(\globalstate {l'},B')$}}
\end{prooftree}
\vspace*{-.5cm}
  \caption{\label{fig:kexchangePROPO} 
  Definition of the relation $\transitionKE{e}{k}$ 
 }
\end{figure}

Finally, let $\system$ be a system and  $\transitionKE{e}{k}$ be the transition relation given in Fig.~\ref{fig:kexchangePROPO}
among abstract configurations of the form $(\globalstate l,B)$.
$\globalstate l$ is a global control state of  $\system$
and $B:\procSet\to\big(2^{\procSet}\times 2^{\procSet}\big)$ is the function  defined above that 
associates to each process $p$ a pair of sets of processes
$B(p)=(C_{S,p},C_{R,p})$. Transition
$\transitionKE{e}{k}$ updates these sets
with respect to the current \kE{k} $e$. 
Causal delivery is verified by checking that for all $ p \in \procSet, p\not\in C_{R,p}'$ meaning that there is no cyclic dependency as stated in Theorem \ref{thm:causal-delivery-graphically}.
%
%
The initial state  is $(\globalstate l_0,B_0)$, where 
$B_0:\procSet\to(2^{\procSet}\times 2^{\procSet})$ denotes the function such that $B_0(p)=(\emptyset,\emptyset)$ for all $p\in \procSet$.

    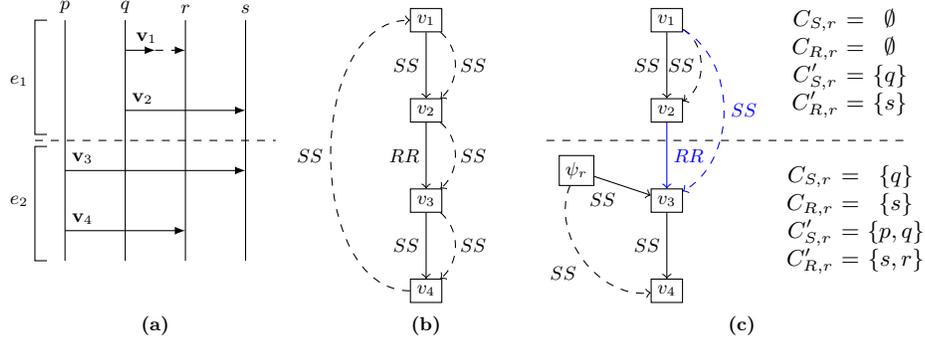
\begin{figure}[t]
\begin{center}
\begin{tikzpicture}[scale = 0.8, every node/.style={scale=0.8}]
\begin{scope}[shift = {(6,-0.5)}]
	\node[draw] (a) at (0,0) {$v_1$};
	\node[draw] (b) at (0,-1.5) {$v_2$};
	\node[draw] (c) at (0,-3) {$v_3$};
	\node[draw] (d) at (0,-4.5) {$v_4$};	
	
	\draw[->] (a) -- (b)node[midway, left]{$SS$};
	\draw[->] (b) -- (c)node[midway, left]{$RR$};
	\draw[->] (c) -- (d)node[midway, left]{$SS$};
	\draw[->, dashed] (a) to[bend left = 50] (b);
	\draw[->, dashed] (b) to[bend left = 50] (c); 
	\draw[->, dashed] (c) to[bend left = 50] (d);
	\draw[->, dashed] (d) to[bend left = 90] (a); 
	
	\draw(-1.6,-2.25) node[left]{$SS$} ;
	\draw(0.45,-0.75) node[right]{$SS$}; 
	\draw(0.45,-2.25) node[right]{$SS$}; 
	\draw(0.45,-3.75) node[right]{$SS$}; 

	\draw(0,-4.8) node[below]{\textbf{(b)}};
\end{scope}  
\begin{scope}
  \coordinate(pa) at (0,-0.5) ;
  \coordinate (pb) at (0,-4.5) ;
  \coordinate (qa) at (1,-0.5) ;
  \coordinate (qb) at (1,-4.5) ;
  \coordinate (ra) at (2,-0.5) ;
  \coordinate (rb) at (2,-4.5) ;
  \coordinate (sa) at (3,-0.5) ;
  \coordinate (sb) at (3,-4.5) ;

  \draw (pa) node[above]{$p$} ;
  \draw (qa) node[above]{$q$} ;
  \draw (ra) node[above]{$r$} ;
  \draw (sa) node[above]{$s$} ;
  \draw (pa) -- (pb) ;
  \draw (qa) -- (qb) ;
  \draw (ra) -- (rb) ;
  \draw (sa) -- (sb) ;
  
  \coordinate (s1) at (1,-1); 
  \coordinate (m1) at (1.5, -1);
  \coordinate (r1) at (2,-1); 
  \draw[>=latex,->] (s1) -- (m1) node[pos=0.7, above] {$\amessage_1$};
	\draw[>=latex,->,dashed] (m1) -- (r1);
  
  \coordinate (s2) at (1, -2); 
  \coordinate (r2) at (3, -2); 
  \draw[>=latex, ->] (s2) -- (r2); 
  \draw (s2) node[above right]{$\amessage_2$};
 
  \coordinate (s3) at (0, -3); 
  \coordinate (r3) at (3, -3); 
  \draw[>=latex, ->] (s3) -- (r3); 
  \draw (s3) node[above right]{$\amessage_3$};

  \coordinate (s4) at (0, -4); 
  \coordinate (r4) at (2, -4); 
  \draw[>=latex, ->] (s4) -- (r4); 
  \draw (s4) node[above right]{$\amessage_4$};
  
		\coordinate (k2a) at (-0.5,-2.5); 
	\coordinate (k2b) at (3.5,-2.5); 
	\draw[dashed] (k2a) -- (k2b);

	\coordinate (acc1b) at (-0.5,-0.5); 
	\coordinate (acc1a) at (-0.3,-0.5); 
	\coordinate (acc1c) at (-0.5,-2.4); 
	\coordinate (acc1d) at (-0.3,-2.4); 
	\draw (acc1a) -- (acc1b) -- (acc1c) -- (acc1d);

	\coordinate (acc2b) at (-0.5,-2.6); 
	\coordinate (acc2a) at (-0.3,-2.6); 
	\coordinate (acc2c) at (-0.5,-4.5); 
	\coordinate (acc2d) at (-0.3,-4.5); 
	\draw (acc2a) -- (acc2b) -- (acc2c) -- (acc2d);

	\draw (-0.5, -1.5) node[left]{ $e_1$};
	\draw (-0.5, -3.5) node[left]{ $e_2$};

\draw(1.5, -5.3) node[below] {\textbf{(a)}};
	
\end{scope} 

\begin{scope}[shift = {(10,0)}]
	\node[draw] (a) at (0,-0.5) {$v_1$};
	\node[draw] (b) at (0,-2) {$v_2$};
	
	\draw[->] (a) -- (b)node[midway, left]{$SS$};
	\draw[->, dashed] (a) to[bend left = 60] node[midway, left]{$SS$} (b);


	\node[text width=4cm,text centered, scale=1.2] at(3.2,-1.2){
	  $\begin{array}{c@{\ =\ }c@{\quad}c@{\ =\ }c}
            C_{S,r} &   \emptyset   \\
             C_{R,r} & \emptyset \\
            C'_{S,r} & \{q\}  \\
            C'_{R,r} & \{s\} \\
	\end{array}$};


	\coordinate (k1a) at (-2,-2.5); 
	\coordinate (k1b) at (4,-2.5); 
	\draw[dashed] (k1a) -- (k1b);
\end{scope} 

\begin{scope}[shift = {(10,-4)}]
	\draw(1.25,-1.3) node[below] {\textbf{(c)}};

	\node[draw] (c) at (0,0.5) {$v_3$};
	\node[draw] (d) at (0,-1) {$v_4$};	
	\node[draw] (e) at (-1.5,1) {$\lambdanode_r$};
	
	\draw[->] (c) -- (d)node[midway, left]{$SS$};
	\draw[->, blue] (b) -- (c)node[midway, right] {$RR$};
	\draw[->, dashed, blue, bend left=60] (a) to node [midway, right] {$SS$}(c);
	\draw[->] (e) to node[midway, below left]{$SS$} (c);
	\draw[->, dashed] (e) to[bend right = 60] node [midway, below left] {$SS$}(d);


	\node[text width=4cm,text centered, , scale=1.2] at(3.3,0.2){
	  $\begin{array}{c@{\ =\ }c@{\quad}c@{\ =\ }c}
            C_{S,r} &   \{q\} \\
            C_{R,r} & \{s\} \\
            C'_{S,r} & \{p,q\} \\	
            C'_{R,r} & \{s,r\}\\
	\end{array}$};
\	
	\end{scope} 
\end{tikzpicture} 
\end{center}
\vspace*{-0.5cm}
\caption{ (a) an MSC 
	(b) its associated global conflict graph, (c) the conflict graphs of its \kE{k}s
	}
\label{fig:exmp-invalid-execution-B}
\end{figure}

\begin{example}[An invalid execution]   
Let $e = e_1 \cdot e_2$ with $e_1$ and $e_2$ the two \kE{2}s of this execution.
such that $e_1 =  \send{q}{r}{\amessage_1}\cdot \send{q}{s}{\amessage_2}\cdot \rec{q}{s}{\amessage_2}$ and $ e_2 =  \send{p}{s}{\amessage_3}\cdot \rec{p}{s}{\amessage_3}\cdot \send{p}{r}{\amessage_4}\cdot \rec{p}{r}{\amessage_4}$.
%
Fig.~\ref{fig:exmp-invalid-execution-B}a and \ref{fig:exmp-invalid-execution-B}c show the MSC and corresponding conflict graph of each of the \kE{2}s.
Note that two edges of the global graph (in blue) ``go across'' \kE{k}s. These edges do not belong to the local conflict graphs and are mimicked
by the incoming and outgoing edges of summary nodes. 
The values of sets $C_{S,r}$ and $C_{R,r}$ at the beginning and at the end of the \kE{k} are given on the right.
All other sets $C_{S,p}$ and $C_{R,p}$ for $p \neq r$ are empty, since there is only one unmatched message to process $r$.
Notice how at the end of the second \kE{k}, $r \in C'_{R,r}$ signalling that message $v_4$ violates causal delivery.
\end{example}

\begin{comp} 
In \cite{DBLP:conf/cav/BouajjaniEJQ18} the authors define
$\transitionKE{e}{k}$ in a rather different way:
they do not explicitly
give a graph-theoretic characterisation of causal delivery; instead they compute,
for every process $p$, the set $B(p)$ of processes that either
sent an unmatched message to $p$ or received a message
from a process in $B(p)$. They then make sure that any message sent to $p$ by
a process $q\in B(p)$ is unmatched.
According to that definition,
  the MSC of Fig.~\ref{fig:exmp-invalid-execution-B}b would
satisfy causal delivery and would be \kSous{1}.
However, this is not the case (this MSC does not satisfy causal delivery) as we have shown in Example~\ref{ex:causal-delivery}. 
Due to to the above errors, we had to propose a considerably different approach. 
The extended edges of the conflict graph, and the
graph-theoretic characterisation of causal delivery as well as summary nodes, have no equivalent in
\cite{DBLP:conf/cav/BouajjaniEJQ18}.
\end{comp} 

%

Next lemma proves that Fig.~\ref{fig:kexchangePROPO} properly characterises causal delivery.

\begin{lemma}\label{lem:causal}
  An MSC $msc$ is \kSous{k} iff
  there is $e~=~e_1\cdots e_n$  a linearisation  
  such that
  $(\globalstate{l_0},B_0)\transitionKE{e_1}{k}\cdots\transitionKE{e_n}{k}(\globalstate{l'},B')$ for some global state $\globalstate{l'}$ and
  some $B':\procSet\to(2^{\procSet}\times 2^{\procSet})$.
\end{lemma}
\iflong
\begin{proof}
\begin{quote}
{\bf Lemma 1}
  An MSC $msc$ is \kSous{k} iff
  there is a linearisation  
  $e~=~e_1\cdots e_n$ such that
  $(\globalstate{l_0},B_0)\transitionKE{e_1}{k}\cdots\transitionKE{e_n}{k}(\globalstate{l'},B')$ for some global state $\globalstate{l'}$ and
  some $B':\procSet\to(2^{\procSet}\times 2^{\procSet})$.
\end{quote} 
 
 \begin{proof}
 $\Rightarrow$ 
 Since  $msc$ is  \kSous{k} then  $\exists e=e_1\cdots e_n$ such that $e$ is a linearisation of $msc$. The proof proceeds by induction on $n$. 
 \begin{description}
 \item [Base case]
 If $n=1$ then $e = e_1$. Thus there is only one \kE{k} and the local  conflict graph $\cgraph{e,B'}$ is the same as the complete global one $\cgraph{e}$.
 By hypothesis, 
 	as $msc$ satisfies causal delivery 
 	we have that  for some $\B$, $(\globalstate{l}_0, B_0) \xRightarrow{e} (\globalstate{l'}, B')$.
   
  By contradiction, suppose that $\exists p \in \procSet$ such that $p \in C'_{R,p}$. 
  Whence there	exists $ v'$ matched, such that $p = \receiver{v'}$ and $v \cgedgeD{SS} v'$ with $v \in \mathsf{Unm}_p$. By Rule 4 (Fig. \ref{fig:conflict-graph-extension}),  an edge $ v' \cgedgeD{SS} v$ has been added to the extended conflict graph.  Thus, there is a cycle $\cgedgeD{SS}$ from $v$ to $v$ and this violates Theorem \ref{thm:causal-delivery-graphically}, which is a contradiction. 
 
  \item [Inductive step] If $n >1$, by inductive hypothesis, we have 
  $$(\globalstate{l_0}, B_0) \transitionKE{e_1}{k} \cdots \transitionKE{e_{n-1}}{k} (\globalstate{l_{n-1}}, B)$$ with $B= (C_{S,p},C_{R,p})_{p \in \procSet}$. 
Since receptions correspond to sends in the current \kE{k} and all sends precede all receptions, we have $(\globalstate{l_{n-1}}, \B_0) \xRightarrow{e} (\globalstate{l_{n}}, \B)$ for some $\B$. 

By inductive hypothesis we have that 
\begin{align*}
C_{S,p} = &\  \{\procofactionv{S}{v'} \mid v \cgedgeD{SS} v' \mbox{ \& } v \mbox{ not matched} \mbox{ \& } \procofactionv{S}{v} = p  \} \\
C_{R,p} =  &\  \{\procofactionv{R}{v'} \mid v \cgedgeD{SS} v' \mbox{ \& } v \mbox{ not matched} \mbox{ \& } \procofactionv{R}{v}= p \mbox{ \& }  v' \cap R \neq \emptyset \}  
\end{align*} 
  By contradiction, suppose that there is a process $ p \in C'_{R,p}$. Then by construction  there exist two nodes $v$ and $v'$ such that $v \cgedgeD{SS} v'$, $v \in \mathsf{Unm}_p$, $v'$ matched and $\receiver{v'} = p$. We can have the following situations:  
  \begin{enumerate}
  \item $v\in V$, then both message exchanges $v$ and $v'$ with $v$ unmatched and $v'$ matched are in the current \kE{k} then we can easily reach a contradiction and the proof proceeds as in the base case.  
  \item $v=\lambdanode_p$, then by inductive hypothesis there exists a non-matched message $v_p \in V$ belonging to a previous \kE{k}. We want to show that if this is the case we can reconstruct a cyclic path in the  extended conflict graph, 
  which is a contradiction.


We assume that by inductive hypothesis 
$\cgraph{e}$ has been reconstructed from the local conflict graphs considering actions in $e_1 \dots e_{n-1}$. We now analyse the last \kE{k} and  describe to what each edge corresponds in $\cgraph{e}$
. There are four cases:
\begin{enumerate}
\item $v_1 \cgedge{XY} v_2$ with $v_1, v_2 \in V$, this edge exists also in $\cgraph{e}$
\item $\lambdanode_q \cgedge{SX} v_1$ with $v_1 \in V$. Then in 
$\cgraph{e}$ there exists an unmatched message $v_q$ and this extra edge has been constructed from  Equations \ref{eq:SS},  \ref{eq:RR} or  \ref{eq:unmatchedSS}:

If $\lambdanode_q \cgedge{SR} v_1 $ then $\procofactionv{R}{v_1} \in C_{S,q}$ thus by inductive hypothesis there exists $v \in V$ in 
$\cgraph{e}$ such that $v_q \cgedgeD{SS} v $ and $\procofactionv{S}{v} = \procofactionv{R}{v_1}$. Whence there exists an edge   $v \cgedgeD{SR} v_1 $ in 
$\cgraph{e}$.
If the edge $\lambdanode_q \cgedge{SS} v_1$ has been added as $\procofactionv{S}{v_1} \in C_{S,q}$ then by inductive hypothesis there exists, in 
$\cgraph{e}$, a node $v$ reachable with an edge $\cgedgeD{SS}$ from $v_q$ such that $\procofactionv{S}{v} = \procofactionv{S}{v_1}$. Thus an edge $v \cgedgeD{SS} v_1 $ exists in 
$\cgraph{e}$.

If the edge $\lambdanode_q \cgedge{SS} v_1$ has been added as $\procofactionv{R}{v_1} \in C_{R,q}$ and $v_1$ is a matched send. Then by inductive hypothesis there exists, in 
$\cgraph{e}$, a node $v$ reachable with an edge $\cgedgeD{SS}$ from $v$ such that $\procofactionv{R}{v} = \procofactionv{R}{v_1}$.
Whence in 
$\cgraph{e}$ there exists an edge $v \cgedge{RR} v_1$. 

If the edge $\lambdanode_q \cgedge{SS} v_1$ has been added as $\procofactionv{R}{v_1} \in C_{R,q}$ and $v_1$ is an unmatched send. Then by inductive hypothesis there exists, in 
$\cgraph{e}$, a node $v$ reachable with an edge $\cgedgeD{SS}$ from $v$ such that $\procofactionv{R}{v} = \procofactionv{R}{v_1}$.
Whence in 
$\cgraph{e}$, because of Rule  (4) in Fig. \ref{fig:conflict-graph-extension} there exists an edge $v \cgedgeD{SS} v_1$. 

If the edge $\lambdanode_q \cgedge{SS} v_1$ has been added as $\procofactionv{S}{v_1} \in C_{R,q}$. Then by inductive hypothesis there exists, in 
$\cgraph{e}$, a node $v$ reachable with an edge $\cgedgeD{SS}$ from $v$ such that $\procofactionv{R}{v} = \procofactionv{S}{v_1}$.
Whence in 
$\cgraph{e}$,  there exists an edge $v \cgedge{RS} v_1$. 

\item $v_1 \cgedge{SS} \lambdanode_q$ with $v_1 \in V$ and $v_1$ matched, then we know $\procofactionv{R}{v_1} = q$ and because of Rule  (4) in Fig. \ref{fig:conflict-graph-extension} in 
$\cgraph{e}$ there exists an edge $v_1 \cgedgeD{SS} v_q$.

\item $\lambdanode_q \cgedge{SS} \lambdanode_r$, thus $r \in C_{R,q}$. This means that there exists a matched message $v$ such that $v_q \cgedge{SS} v$ and $\procofactionv{R}{v} = r $. Thus, in 
$\cgraph{e}$,  we can add,  because of Rule  (4) in Fig. \ref{fig:conflict-graph-extension}, the edge $v \cgedgeD{SS} v_r$.

\end{enumerate}

Then it follows that if there exists an edge  $\lambdanode_p \cgedgeD{SS} v'$ it means that an edge $v_p \cgedgeD{SS} v'$ exists in the global extended conflict graph and thus by applying Rule (4) in Fig. \ref{fig:conflict-graph-extension} we can reach a contradiction, as we have a cycle.
\end{enumerate}

  \end{description}

  $\Leftarrow$
If $e = e_1 \cdots e_n$, where each $e_i$ corresponds to a valid \kE{k}. Suppose by contradiction that  $msc(e)$ 
violates causal delivery. By Theorem \ref{thm:causal-delivery-graphically} then the global extended conflict graph must contain an edge 
 $v \cgedgeD{SS} v$. This means that there is an unmatched message $v_p$ to process $p$ that is causally followed by a matched message $v$ to the same process $p$.
 Since each $e_i$ is a valid \kE{k} we know that such an edge cannot appear in any of the local conflict graphs. 
Indeed, if such an edge existed then there should be an edge $\cgedgeD{SS}$ from  $v_p$ or $\lambdanode_p$ (if the two messages belong to two different \kE{k}s) to $v$. But in this case we would have $p \in C'_{R,P}$ which is a contradiction.\qed
\end{proof}
\end{proof}
\fi

  Note that there are only finitely many abstract configurations of the form
  $(\globalstate{l},B)$ with $\globalstate{l}$ a tuple of control states and
  $B:\procSet \to (2^{\procSet}\times 2^{\procSet})$.
  Moreover, since $\paylodSet$ is  finite, the alphabet over the possible \kE{k} for a given $k$ is also finite.
  Therefore $\transitionKE{e}{k}$ is a relation
  on a finite set, 
and the set $sTr_k(\system)$ of \kSous{k} MSCs of a   
  system $\system$ forms a regular language.  It follows that  it is decidable
  whether a given abstract configuration of the form $(\globalstate{l}, B)$ is reachable from the initial configuration following a \kSable{k} execution.

\begin{theorem}\label{thm:reachability-is-decidable}
  Let $\system$ be a \kSable{k} system and $\globalstate{l}$ a global
  control state of $\system$. The problem whether there
  exists $e\in\asEx(\system)$ and $\B$ such that
  $(\globalstate{l_0},\B_0)\xRightarrow{e}(\globalstate{l},\B)$
  is decidable.  
\end{theorem}

%

\begin{remark}
  Deadlock-freedom, unspecified receptions, and
  absence of orphan messages are other properties that become
  decidable for a \kSable{k} system because 
  of the regularity of
the set of \kSous{k} MSCs. 
\end{remark}


\section{Decidability of \kSity{k} for mailbox systems }
\label{sec:dec} 
 
We establish, here, the decidability of \kSity{k};
 our approach is similar to the one of 
Bouajjani~\emph{et al.} based on the notion of borderline violation, but we adjust it to adapt to the new characterisation of \kSable{k} executions (Theorem \ref{thm:k-sity-scc}).

\begin{definition}[Borderline violation]
  A non \kSable{k} execution $e$ 
  is a borderline violation if $e=e'\cdot r$, $r$ is a reception and  $e'$ is \kSable{k}.
\end{definition}
 
Note that a system $\system$ that is not \kSable{k} always admits at least
one borderline violation $e'\cdot r\in \asEx(\system)$ with $r\in\receiveSet$:
indeed, there is at least one execution
$e\in\asEx(\system)$ 
which contains a unique
minimal prefix of the form $e'\cdot r$ that is not \kSable{k}; moreover
since $e'$ is \kSable{k}, $r$ cannot be a \kE{k} of just one send action,
therefore it must be a receive action.
In order to find such a borderline violation, Bouajjani~\emph{et al.} introduced an 
instrumented system $\system'$ that behaves like $\system$,
except that it contains an extra process $\pi$,
and such that a non-deterministically chosen  message that  should
have been sent from a process $p$ to a process $q$ may now be sent from $p$
to $\pi$, and later forwarded by $\pi$ to $q$.
In $\system'$, each process $p$ has the possibility, instead of sending a message
$\amessage$ to $q$, to
deviate this message to $\pi$; if it does so, $p$ continues its execution
as if it really had sent it to $q$. Note also that the message sent to
$\pi$ get tagged with the original destination process $q$.
Similarly, for each possible reception, a process has the possibility
to receive a given message not from the initial sender but from $\pi$. 
The process $\pi$ has an initial state from which it can receive any messages from the system. Each reception makes it go into a different state. From this state, it is able to send the message back to the original recipient. Once a message is forwarded, $\pi$ reaches its final state and remains idle. The precise definition of the instrumented system can be found in  Appendix~\ref{sec:additional}. The following example illustrates how the instrumented system works.
 \begin{example}[A deviated message]
 
\hspace*{-0.63cm} 
\begin{minipage}[b]{8.5cm}
Let $e_1$, $e_2$ be two executions of a system $\system$ with MSCs respectively $msc(e_1)$ and $msc(e_2)$. $e_1$  is not \kSable{1}. It is   borderline in $\system$. If we delete the last reception,
it becomes indeed \kSable{1}. $msc(e_2)$ is the MSC obtained from the
instrumented system $\system'$ where the message $\amessage_1$ is first deviated to $\pi$ and then sent back to $q$ from $\pi$.

Note that
$msc(e_2)$ is \kSous{1}.  In this case, the instrumented system $\system'$
in the \kSous{1} semantics ``reveals'' the existence of
a borderline violation of $\system$.
 \end{minipage}
\hfill
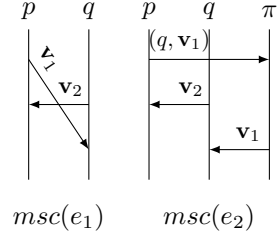
\begin{minipage}[b]{4.5cm}

\begin{tikzpicture}[scale = 0.8 ]
\begin{scope}
  \draw(0.5,-3.5) node[above]{$msc(e_1)$};
  \coordinate (pa) at (0,0) ;
  \coordinate (pb) at (0,-2.5) ;
  \coordinate (qa) at (1,0) ;
  \coordinate (qb) at (1,-2.5) ;
  \draw (0,0.3) node{$p$} ;
  \draw (1,0.3) node{$q$} ;
  \draw (pa) -- (pb) ;
  \draw (qa) -- (qb) ;
  
  \coordinate (sa) at (0,-0.5); 
  \coordinate (ra) at (1,-2); 
  \draw[>=latex, ->] (sa) to node[above, sloped, pos = 0.1, scale = 0.9]{$\amessage_1$} (ra); 
  
  \coordinate (sb) at (1, -1.25); 
  \coordinate (rb) at (0, -1.25); 
  \draw[>=latex, ->] (sb) to node[above,pos = 0.3, scale = 0.9]{$\amessage_2$} (rb); 

\end{scope}
\begin{scope}[shift = {(2,0)}]
  \draw(1,-3.5) node[above]{$msc(e_2)$};
  \coordinate(pa) at (0,0) ;
  \coordinate (pb) at (0,-2.5) ;
  \coordinate (qa) at (1,0) ;
  \coordinate (qb) at (1,-2.5) ;
  \coordinate (pia) at (2,0) ;
  \coordinate (pib) at (2,-2.5) ;
  \draw (0,0.3) node{$p$} ;
  \draw (1,0.3) node{$q$} ;
  \draw (2,0.3) node{$\pi$} ;
  \draw (pa) -- (pb) ;
  \draw (qa) -- (qb) ;
  \draw (pia) -- (pib) ;
  
  \coordinate (sa) at (0,-0.5); 
  \coordinate (ra) at (2,-0.5); 
  \draw[>=latex, ->] (sa) to node[above, pos = 0.25, scale = 0.85]{$(q,\amessage_1)$} (ra); 
  \draw (sa) ;
  
  \coordinate (sb) at (1, -1.25); 
  \coordinate (rb) at (0, -1.25); 
  \draw[>=latex, ->] (sb) to node[above,pos = 0.3, scale = 0.9]{$\amessage_2$} (rb);

  \coordinate (se) at (2, -2); 
  \coordinate (re) at (1, -2); 
  \draw[>=latex, ->] (se) -- (re); 
  \draw (se) node[above left, scale = 0.9]{$\amessage_1$};
 
\end{scope}
\end{tikzpicture}

\end{minipage}
\end{example}


%

For each execution $e\cdot r\in \asEx(\system)$ that ends with a reception,
there exists an execution $\deviate{e\cdot r}\in\asEx(\system')$ where
the message exchange associated with the reception $r$ has been deviated to
$\pi$; formally, if $e\cdot r = e_1\cdot s\cdot e_2\cdot r$ with
$r=\rec{p}{q}{\amessage}$ and $s \matches r$, then
$$
\deviate{e\cdot r}=
e_1\cdot\send{p}{\pi}{(q,\amessage)}\cdot \rec{p}{\pi}{(q,\amessage)}\cdot
e_2 \cdot \send{\pi}{q}{(\amessage)}\cdot \rec{\pi}{q}{\amessage}.
$$

\begin{definition}[Feasible execution, bad execution]
  A \kSable{k} exe\-cution $e'$ of $\system'$ is \emph{feasible} if there
  is an execution $e\cdot r\in\asEx(\system)$ such that
  $\deviate{e\cdot r}=e'$.
  A feasible execution $e' = \deviate{e\cdot r}$ of $\system'$ is \emph{bad} if execution $e\cdot r$ is not \kSable{k} in $\system$. 
\end{definition}


\hspace*{-0.63cm}
\begin{minipage}[b]{8.5cm}
\begin{example}[A non-feasible execution]

  Let $e'$ be an execution
  such that $msc(e')$ is as depicted on the right.
  Clearly, this MSC satisfies causal delivery and could be the execution
  of some instrumented system $\system'$.
  However, the sequence $e\cdot r$
  such that $\deviate{e\cdot r}=e'$ does not satisfy causal delivery,
  therefore it cannot be an execution of the original system $\system$.
  In other words, the execution $e'$ is not feasible. 
    \end{example}

    \end{minipage}
   \hfill
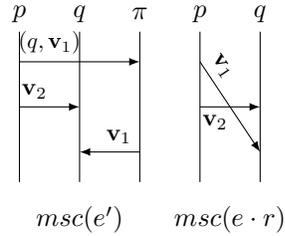
\begin{minipage}[t]{4.5cm}
        \begin{tikzpicture}[scale = 0.8]
          \begin{scope}
	    \draw(1,-3.5) node[above]{$msc(e')$};
	    \coordinate (pa) at (0,0) ;
	    \coordinate (pb) at (0,-2.5) ;
	    \coordinate (qa) at (1,0) ;
	    \coordinate (qb) at (1,-2.5) ;
	    \coordinate (ra) at (2,0) ;
	    \coordinate (rb) at (2,-2.5) ;
	    
	      \draw (0,0.3) node{$p$} ;
 		 \draw (1,0.3) node{$q$} ;
	     \draw (2,0.3) node{$\pi$} ;

	    \draw (ra) -- (rb) ;
	    \draw (pa) -- (pb) ;
	    \draw (qa) -- (qb) ;
	    
	    \coordinate (sa) at (0,-0.5); 
	    \coordinate (ra) at (2,-0.5); 
	    \draw[>=latex, ->] (sa) to node[sloped, above, scale = 0.85, pos= 0.25]{$(q,\amessage_1)$} (ra); 
	    
	    \coordinate (sb) at (0, -1.25); 
	    \coordinate (rb) at (1, -1.25); 
	    \draw[>=latex, ->] (sb) to node[sloped, above, pos = 0.25, scale = 0.9]{$\amessage_2$}  (rb); 
	    
	    \coordinate (sc) at (2, -2); 
	    \coordinate (rc) at (1, -2); 
	\draw[>=latex, ->] (sc) -- (rc); 
	\draw (sc) node[above left, scale = 0.9]{$\amessage_1$};
          \end{scope}
          \begin{scope}[shift = {(3,0)}]
	    \draw(0.5,-3.5) node[above]{$msc(e\cdot r)$};
	    \coordinate(pa) at (0,0) ;
	    \coordinate (pb) at (0,-2.5) ;
	    \coordinate (qa) at (1,0) ;
	    \coordinate (qb) at (1,-2.5) ;
	
  \draw (0,0.3) node{$p$} ;
  \draw (1,0.3) node{$q$} ;
	    \draw (pa) -- (pb) ;
	    \draw (qa) -- (qb) ;
	    
	    \coordinate (sa) at (0,-0.5); 
	    \coordinate (ra) at (1,-2); 
	    \draw[>=latex, ->] (sa) to node[pos=0.15, sloped, above, scale = 0.9]{$\amessage_1$}  (ra); 
	    
	    \coordinate (sb) at (0, -1.25); 
	    \coordinate (rb) at (1, -1.25); 
	    \draw[>=latex, ->] (sb) to node[pos=0.25, sloped, below, scale = 0.9]{$\amessage_2$} (rb); 
	        
          \end{scope}
        \end{tikzpicture} 
    \end{minipage}

\begin{lemma}
\label{lem:bad-is-bad}
  A system $\system$ is not \kSable{k} iff there is a \kSable{k}
  execution $e'$\! of $\system'$ that is feasible and bad.
\end{lemma}
 \iflong
\begin{proof}

\begin{quote}
{\bf Lemma 2}
  A system $\system$ is not \kSable{k} iff there is a \kSable{k}
  execution $e'$\! of $\system'$ that is feasible and bad.
\end{quote}

\begin{proof}
$\Rightarrow$
Let $\system$ be not \kSable{k} then there exists an execution that is not \kSable{k} which contains a unique minimal prefix of the form $e \cdot r$ with $e$ \kSable{k} and $r = \rec{p}{q}{\amessage}$ a receive action. Thus $e$ is bad and there exists an $e' = \deviate{e\cdot r} \in \asEx(\system')$.

Since $e$ is \kSable{k}, 
$msc(e)$ is \kSous{k} and there exists a linearisation $e''$ such that 
$e''=e_1\dots e_n$ and there exists a \kE{k} $e_i$ containing the send action  $\send{p}{q}{\amessage}$. Now if we replace this action with $\send{p}{\pi}{(q,\amessage)}$ and we add at the end of the same \kE{k} the action $\rec{p}{\pi}{(q,\amessage)}$. The execution in $\asEx(\system')$ remains \kSable{k}. 
Finally if we add to $e''$ a new \kE{k} with the actions $\send{\pi}{q}{\amessage}$ and $\rec{\pi}{q}{\amessage}$ the execution remains \kSable{k}, hence $e'$ is feasible.

$\Leftarrow$
If there is a \kSable{k}
  execution $e'$ of $\system'$ that is feasible and bad, then by construction $e' = \deviate{e\cdot r}$ and $e\cdot r$ is not \kSable{k}. Whence $\system$ is not \kSable{k} and this concludes the proof. \qed
  
  \end{proof}
\end{proof}
\fi

As we have already noted, the set
of \kSous{k} MSCs of $\system'$ is regular. 
The decision procedure
for \kSity{k} follows from the fact that the set of
MSCs that have as linearisation a feasible bad execution
as we will see,
is regular as well, and that it can be
recognised by an (effectively computable)
non-deterministic finite state automaton. The decidability of \kSity{k}
follows then from Lemma~\ref{lem:bad-is-bad} and the decidability of the
emptiness problem for non-deterministic finite state automata.

\paragraph{\bf Recognition of feasible executions.} 
We start with 
the automaton that recognises feasible executions; for this,
we revisit the construction we just used for recognising sequences of
\kE{k}s that satisfy causal delivery.

In the remainder, we assume an execution $e'\in \asEx(\system')$
that contains exactly one send of the form
$\send{p}{\pi}{(q,\amessage)}$ and one reception 
of the form $\rec{\pi}{q}{\amessage}$, this reception being the last action of $e'$.
Let $(V,\{\cgedge{XY}\}_{X,Y\in\{R,S\}})$
be the conflict graph of $e'$.  There are two uniquely determined
vertices $\vstart,\vstop\in V$ such that $\receiver{\vstart}=\pi$
and $\sender{\vstop}=\pi$ that correspond, respectively, to the first
and last message exchanges of the deviation. The conflict graph
of $e\cdot r$ is then obtained by merging these two nodes.

\begin{lemma}\label{lem:graph-carac-feas}
The execution $e'$ is not feasible iff
  there is a vertex $v$ in
      the conflict graph of $e'$ such that
      $\vstart\cgedgeD{SS} v\cgedge{RR}\vstop$.
\end{lemma}

\iflong
\begin{proof}

\begin{quote}
{\bf Lemma 3}
  The execution $e'$ is not feasible iff
  there is a vertex $v$ in
      the conflict graph of $e'$ such that
      $\vstart\cgedgeD{SS} v\cgedge{RR}\vstop$.
\end{quote}

\begin{proof}

$\Leftarrow$ 
If there is $v$ such that $\vstart\cgedgeD{SS}v\cgedge{RR}\vstop$,
  this means that a message sent after the deviated message is received before
  it: hence, it violates causal delivery.

  $\Rightarrow$ Assume now $e'$ not feasible this entails that $e$ is \kSable{k} and $e\cdot r$ violates causal
  delivery. Whence there exists an
  unmatched message that becomes matched and because of  Definition~\ref{def:causal-delivery}  
  there are $i',j'$ such that $r=a_{i'}$ $a_i\matches a_{i'}$, $a_j\matches a_{j'}$, and $i \prec j$ and $j'\prec i'$.
  So the conflict graph $\cgraph{e\cdot r}$ contains two vertices
  $v_d=\{a_i, a_{i'}\}$ and $v\{a_j, a_{j'}\}$ such that $v_d\cgedgeD{SS}v\cgedge{RR}v_d$. Because of the deviation  node $v_d$ is split in nodes $\vstart$
  and $\vstop$ in $\cgraph{e'}$, and therefore we conclude that 
  $\vstart\cgedgeD{SS}v\cgedge{RR}\vstop$. \qed
\end{proof}
\end{proof}
\fi

In order to decide whether an execution $e'$ is feasible, we want 
to forbid that a send action $\send{p'}{q}{\amessage'}$
that happens causally after  $\vstart$ 
is
matched by a receive $\rec{p'}{q}{\amessage'}$ that happens causally
before the reception $\vstop$. 
As a matter of fact, this boils down to deal with the deviated send action as an unmatched send.
So we will consider sets
of processes $C_{S}^{\pi}$ and $C_R^{\pi}$ similar to the ones
used for $\transitionKE{e}{k}$, but with the goal of computing which
actions happen causally after the send to $\pi$. We also introduce a summary
node $\pinode$ and the extra edges following the same principles
as in the previous section.
Formally, let $B:\procSet\to(2^{\procSet}\times 2^{\procSet})$,
$C_{S}^{\pi},C_R^{\pi}\subseteq \procSet$ and $e\in S^{\leq k}R^{\leq k}$ be fixed, and let $\cgraph{e,B}=(V',E')$ be the constraint graph with summary nodes
for unmatched sent messages as defined in the
previous section. The local constraint graph $\cgraph{e,B,C_{S}^{\pi},C_R^{\pi}}$ is defined
as the graph $(V'',E'')$ where $V''=V'\cup\{\pinode\}$ and $E''$ is
$E'$ augmented with
\begin{align*}
  & \{\pinode \cgedge{SX} v \mid\procofactionv{X}{v} \in C_{S}^{\pi} 
  \mbox{ \& } v\cap X\neq\emptyset \mbox{ for some }X\in\{S,R\}\} 
  \\ \cup \  &
 \{\pinode \cgedge{SS} v\mid \procofactionv{X}{v} \in C_{R}^{\pi} 
 \mbox { \& } v \cap R \neq \emptyset \mbox{ for some }X\in\{S,R\}\} 
  \\ \cup \ & 
  \{ \pinode \cgedge{SS} v \mid  \procofactionv{R}{v} \in C_{R}^{\pi}
  \mbox{ \& }  v \mbox{ is unmatched} \}  
\  \cup \ 
\{ \pinode \cgedge{SS} \lambdanode_p \mid p \in C_{R}^{\pi}  \}  
\end{align*}

As before, we consider the ``closure''
$\cgedgeD{XY}$ of these edges by the rules of Fig.~\ref{fig:rules-cgedgeD}.
The transition relation $\transitionCD{e}{k}$
is defined in Fig.~\ref{fig:causalkexchge}.
It relates abstract configurations of the form
$(\globalstate{l},B,\vec C,\destpi)$ with $\vec C=(C_{S,\pi},C_{R,\pi})$ and
$\destpi\in\procSet\cup\{\bot\}$ storing to whom the
message deviated to $\pi$ was supposed to be delivered.
Thus, the initial abstract configuration is
$(l_0, B_0, (\emptyset,\emptyset),\bot)$, where $\bot$ means that the processus
$\destpi$ has not been determined yet. It will be set as soon as the send to process $\pi$ is encountered.

%
 
\begin{figure}[t]
  \begin{prooftree}
\AxiomC{\small
		{\stackanchor
			{  $ (\globalstate{l},B) \transitionKE{e}{k} (\globalstate{l'},B') \qquad e =a_1\cdots a_n \qquad (\forall v)\ \procofaction{S}{v}\neq\pi$}
			{\stackanchor
				{  $(\forall v,v')\ \receiver{v}=\receiver{v'}=\pi\implies v=v'\wedge\destpi=\bot$}
				{\stackanchor 
					{ $ (\forall v)\ v\ni\send{p}{\pi}{(q,\amessage)} \implies \finddest'=q \quad  \finddest\neq\bot \implies \finddest'=\finddest
$}
					{\stackanchor
						{ ${C_{X}^{\pi}}'=  C_{X}^{\pi}\cup
    \{\procofactionv{X}{v'} \mid v \cgedgeD{SS} v' \mbox{ \& } v'\cap X\neq\emptyset\mbox{ \& }(\procofactionv{R}{v}=\pi\mbox{ or } v=\pinode)\}
    $}
    				  	{\stackanchor
    				  		{\stackanchor
    				  			{ $\cup ~ \{\procofactionv{S}{v}\mid\procofactionv{R}{v}=\pi\mbox{ \& } X=S\}$}
    				  			{\stackanchor{ $\cup~\{p\mid p\in C_{X,q}\mbox{ \& } v\cgedgeD{SS}\lambdanode_{q}  \mbox{ \& }(\procofactionv{R}{v}=\pi\mbox{ or } v=\pinode)$\}}
    				  			{ $  \finddest'\not\in {C_{R}^{\pi}}'$}
    				  			}
    				  		}{}
  							}}}}}}
 \UnaryInfC{{\small $  (\globalstate{l},B,C_S^{\pi},C_R^{\pi},\destpi)
  \transitionCD{e}{k}
  (\globalstate{l'},B',{C_S^{\pi}}',{C_R^{\pi}}',\destpi')$
}}
  \end{prooftree}
  \vspace*{-0.5cm}
  \caption{\label{fig:causalkexchge}Definition of the relation $\transitionCD{e}{k}$ \vspace*{-0.5cm}}
\end{figure}

%

\begin{lemma}\label{lem:feasible-is-regular}

Let $e'$ be an execution of $\system'$. Then $e'$ is a \kSable{k} feasible execution iff there are $e'' = e_1\cdots e_n \cdot \send{\pi}{q}{\amessage}\cdot \rec{\pi}{q}{\amessage}$ with
  $e_1,\ldots ,e_n\in S^{\leq k}R^{\leq k}$,
$B':\procSet\to2^{\procSet}$, 
  $\vec C'\in (2^{\procSet})^2$,
and a tuple of control states $\globalstate{l'}$ such that $msc(e') = msc(e'')$, $\pi\not\in C_{R,q}$ (with $B'(q)=(C_{S,q},C_{R,q})$), and
$$
(\globalstate{l_0}, B_0, (\emptyset,\emptyset),\bot)
\transitionCD{e_1}{k}\dots\transitionCD{e_n}{k}
(\globalstate{l'}, B', \vec{C'}, q).
$$
\end{lemma}

\iflong
\begin{proof}
  \begin{quote}
{\bf Lemma 4}
	Let $e'$ an execution of $\system'$. Then $e'$ is a \kSable{k} feasible execution iff there are $e'' = e_1\cdots e_n \cdot \send{\pi}{q}{\amessage}\cdot \rec{\pi}{q}{\amessage}$ with
  $e_1,\ldots ,e_n\in S^{\leq k}R^{\leq k}$,
%
$B':\procSet\to2^{\procSet}$, 
  $\vec C'\in (2^{\procSet})^2$,
and a tuple of control states $\globalstate{l'}$ such that $msc(e') = msc(e'')$, $\pi\not\in C_{R,q}$ (with $B'(q)=(C_{S,q},C_{R,q})$), and
$$
(\globalstate{l_0}, B_0, (\emptyset,\emptyset),\bot)
\transitionCD{e_1}{k}\dots\transitionCD{e_n}{k}
(\globalstate{l'}, B', \vec{C'}, q).
$$
\end{quote}

\begin{proof}
Let us first state what are the properties of the variables $B, \vec C$ and $\destpi$.

Let $e'$ a \kSable{k} execution of $\system'$ and $e''= e_1\cdots e_n$ such that $msc(e') = msc(e'')$ be fixed, and assume 
that there are $B,\vec C,\destpi$ such that
$$
(\globalstate{l_0}, B_0, \emptyset,\emptyset, \bot)
\transitionCD{e_1}{k}\dots\transitionCD{e_n}{k}
(\globalstate{l}, B, C_S^{\pi},C_R^{\pi},\finddest).
$$
Notice that $\cgraph{e'} = \cgraph{e''}$.
By induction on $n$, we want to establish that
\begin{enumerate}
\item $\finddest=q$ if and only if a message of the form $(q,\amessage)$ was
  sent to $\pi$ in $e'$;
\item there is at most one message sent to $\pi$ in $e'$;
\item let $\vstart$ denote the unique vertex in $\cgraph{e'}$ (if it exists)
  such that $\procofactionv{R}{\vstart}=\pi$; for all $X\in\{S,R\}$,
  $$
  C_X^{\pi}=\{\procofactionv{X}{v}\mid (v\cap X\neq\emptyset \mbox{ \& } 
  \vstart\cgedgeD{SS}v\mbox{ in }\cgraph{e'}) \mbox{ or } (v,X)=(\vstart,S)\}.
  $$
\end{enumerate}
The  first two points easily follow from the definition of
$\transitionCD{e}{k}$. Let us focus on the last point.
The case $n=1$ is immediate. Let us assume that
$$
(\globalstate{l_0}, B_0, (\emptyset,\emptyset),\bot)
\transitionCD{e_1}{k}\dots\transitionCD{e_{n-1}}{k}
(\globalstate{l}, B, \vec{C}, \destpi)
\transitionCD{e_n}{k}
(\globalstate{l'}, B', \vec{C'}, \destpi')
$$
with $C_X^{\pi}=\{\procofactionv{X}{v}\mid (v\cap X\neq\emptyset \mbox{ \& } \vstart\cgedgeD{SS}v\mbox{ in }\cgraph{e_1\cdots e_{n-1}})\mbox{ or } (v,X)=(\vstart,S)\}$ and let us
show that ${C_X^{\pi}}'=\{\procofactionv{X}{v}\mid (v\cap X\neq\emptyset \mbox{ \& } \vstart\cgedgeD{SS}v \mbox{ in } \cgraph{e_1\cdots e_n}\mbox{ or } (v,X)=(\vstart,S)\}$.
\begin{itemize}
\item Let $X\in\{S,R\}$ and $p\in {C_X^{\pi}}'$ and let us show that there is
  some $v$ such that $p=\procofactionv{X}{v}$ and either
  $\vstart\cgedgeD{SX}v$ in $\cgraph{e_1\cdots e_n}$ or $(v,X)=(\vstart,S)$. We reason
  by case analysis on the reason why $p\in{C_X^{\pi}}'$, according
  to the definition of ${C_X^{\pi}}'$ in Fig.~\ref{fig:causalkexchge}.
  \begin{itemize}
  \item $p\in {C_X^{\pi}}$. Then by induction hypothesis there is
    $v$ such that $p=\procofactionv{X}{v}$, and
    $\vstart\cgedge{SS}v$ in $\cgraph{e_1\cdots e_{n-1}}$, and therefore also
    in $\cgraph{e_1\cdots e_{n}}$, or $(v,X)=(\vstart,S)$.
  \item $p=\procofactionv{X}{v'}$, $v\cgedgeD{SS}v'$, $v'\cap X\neq\emptyset$,
    and
    $\procofactionv{R}{v}=\pi$,
    for some message exchanges $v,v'$ of $e_n$.
    Since $\procofactionv{R}{v}=\pi$, $v=\vstart$.
    This shows this case.
  \item $p=\procofactionv{X}{v'}$, $\pinode\cgedgeD{SS}v'$,
    $v'\cap X\neq\emptyset$, for some message exchange $v'$ of $e_n$.
    It remains to show that $\vstart\cgedgeD{SS}v'$.
    From $\pinode\cgedgeD{SS}v'$, there are some $v,Y$ such that $\pinode\cgedge{SY}v$
    in $\cgraph{e_n,B,\vec C}$, $v\cap Y\neq\emptyset$
    and either $v\cgedgeD{YS}v'$ or $(v,Y)=(v',S)$.
    We reason by case analysis on the construction of the edge
    $\pinode\cgedge{SY} v$.
    \begin{itemize}
    \item $\procofactionv{Y}{v}\in C_S^{\pi}$ and $v\cap Y\neq\emptyset$.
      Let $q=\procofactionv{Y}{v}$. Since $q\in C_S^{\pi}$,
      by induction hypothesis there is $v_1$ in a previous \kE{k} such that
      $\vstart\cgedgeD{SS}v_1$ in $\cgraph{e_1\cdots e_{n-1}}$ or $v_1=\vstart$.
      Since
      $\procofactionv{S}{v_1}=\procofactionv{Y}{v}$, there is an edge
      $v_1\cgedge{SY}v$ in $\cgraph{e_1\cdots e_n}$.
      By hypothesis, we also have either $v\cgedgeD{YS}v'$ or $(v,Y)=(v',S)$.
      So in both cases we get $\vstart\cgedgeD{SS}v_1\cgedge{SS}v'$,
      or $\vstart\cgedgeD{SS}v'$ when $v_1=\vstart$. 
    \item $\pinode\cgedge{SS}v$, $\procofactionv{Y}{v}\in C_R^{\pi}$ and $v\cap R\neq\emptyset$.
      Again by induction hypothesis, we have $v_1$ such that $\vstart\cgedgeD{SS}v_1\cgedge{RY}v$,
      therefore $\vstart\cgedgeD{SS}v'$.
    \item $\pinode\cgedge{SS}v$, $\procofactionv{R}{v}\in C_R^{\pi}$ and $v$ unmatched.
      Again by induction hypothesis, we have $v_1$ such that $\vstart\cgedgeD{SS}v_1\cgedge{RS}v$.
      If $v=v'$, we have $\vstart\cgedgeD{SS} v'$, which closes the case. Otherwise, from
      $v\cgedgeD{YS}v'$ and $v$ unmatched we deduce $v\cgedge{SS} v'$; finally
      we  $\vstart\cgedgeD{SS}v_1\cgedge{RS}v\cgedgeD{SS}v'$, so $\vstart\cgedgeD{SS}v'$, which
      closes the case as well. 
    \item $v=\lambdanode_q$ for some $q\in C_R^{\pi}$. Since $\lambdanode_q$ does not have outgoing edges
      of the form $RS$, $\lambdanode_q\cgedgeD{SS} v'$.
      From $q\in C_R^{\pi}$, we get by induction
      hypothesis some node $v_1$ such that $\vstart\cgedgeD{SS}v_1$ and
      $\procofactionv{R}{v_1}=q$.
      As seen in the proof of Lemma~\ref{lem:causal},
      $\lambdanode_q\cgedgeD{SS}v'$ implies that
      there is a vertex $v_2$ from a previous \kE{k} that is an unmatched send to $q$ such
      that $v_2\cgedgeD{SS} v'$ in $\cgraph{e_1\cdots e_n}$. Since $v_1$ is a matched send to $q$ and $v_2$ is an unmatched
      send to $q$, by rule 4 in Fig.~\ref{fig:rules-cgedgeD}
      $v_1\cgedgeD{SS}v_2$. All together,
      $\vstart\cgedgeD{SS}v_1\cgedgeD{SS}v_2\cgedgeD{SS}v'$, which closes this case.
    \end{itemize}    
  \item $p=\procofactionv{X}{v}$, $\procofactionv{R}{v}=\pi$, and $X=S$. Then $v=\vstart$,
    which closes this case.
  \end{itemize}

\item Conversely, let us show that
  for all $X\in\{S,R\}$ and
  $v$ such that $\vstart\cgedgeD{SS}v$ in $\cgraph{e_1\cdots e_n}$, $\procofaction{S}{v}\neq\pi$, and $v\cap X\neq\emptyset$,
  it holds that $\procofactionv{X}{v}\in {C_X^{\pi}}'$ (the corner case to be proved, $(v,X)=(\vstart, S)$, is treated in the last item).
  Again, we reason by induction on the number $n$ of \kE{k}s. If $n=0$, it is immediate as there are no such $v,X$.
  Let us assume that the property holds for all choices of $v_1,X_1$ such that
  $\vstart\cgedgeD{SS}v_1$ in $\cgraph{e_1\cdots e_{n-1}}$, $\procofactionv{S}{v_1}\neq\pi$, and $v_1\cap X_1\neq\emptyset$. Let $v,X$ be fixed with
  $\vstart\cgedgeD{SS}v$ in $\cgraph{e'_1\cdots e'_n}$, and $v\cap X\neq\emptyset$, and let us show that
  $\procofactionv{X}{v}\in {C_X^{\pi}}'$. We reason by case analysis on the occurrence in $e_n$, or not, of both $\vstart$ and $v$.
  \begin{itemize}
  \item $\vstart$ and $v$ are in $e_n$.
    Then from  $\vstart\cgedgeD{SS}v$ in $\cgraph{e_1\cdots e_n}$
    and the proof of Lemma~\ref{lem:causal}, we get that $\vstart\cgedgeD{SS}v$ in $\cgraph{e_n,B}$.
    By definition of  ${C_S^{\pi}}'$ (first line), it contains $\procofactionv{X}{v}$ 
  \item $\vstart$ in $e_n$ and $v$ in $e_1\cdots e_{n-1}$. Then there are $v_1,v_2,q$ such that
    \begin{itemize}
    \item $v_1$ is in $e_n$, and either $\vstart\cgedgeD{SS} v_1$ in $\cgraph{e_1\cdots e_n}$ or $v_1=\vstart$,
    \item $v_2$ is in $e_1\cdots e_{n-1}$,
      $v_1\cgedgeD{SS}v_2$ by rule 4 of Fig.~\ref{fig:rules-cgedgeD}, i.e., $v_1$ is a matched send to $q$ and $v_2$ is an unmatched send to $q$ 
    \item either $v_2\cgedgeD{XS}v$ in $\cgraph{e_1\cdots e_{n-1}}$, or $v_2=v$
    \end{itemize}
    From the first item, by the proof of  Lemma~\ref{lem:causal}, we get either $\vstart\cgedgeD{SS}v_1$ in $\cgraph{e_n,B}$ or $v=v_1$.
    From the second item, we get $v_1\cgedge{SS}\lambdanode_q$ in $\cgraph{e_n,B}$.
    From these two, we get $\pinode\cgedgeD{SS}\lambdanode_p$ in $\cgraph{e_n,B,\vec C}$.
    By definition of $C_X^{\pi}$, we therefore have $C_{X,q}\subseteq C_X^{\pi}$.
    From the third item, we get $\procofaction{X}{v}\in C_{X,q}$. So finally $\procofaction{X}{v}\in C_X^{\pi}$.
  \item $\vstart$ in $e_1\cdots e_{n-1}$ and $v$ in $e_n$.
    Then there are $v_1,v_2,Y,Z$ such that
    \begin{itemize}
    \item either $\vstart \cgedgeD{SY} v_1$ in $\cgraph{e_1\cdots e_{n-1}}$, or $(v,S)=(v_1,Y)$
    \item $v_1\cgedge{YZ}v_2$
    \item either $v_2\cgedgeD{ZS} v$ in $\cgraph{e_1\cdots e_n}$, with both $v_2$ and $v$ in $e_n$, or $(v_2,Z)=(v,S)$
    \end{itemize}
    From the first item, by induction hypothesis, we get
    $\procofaction{Y}{v_1}\in C_X^{\pi}$.
    From the second item, we get $\procofaction{Y}{v_1}=\procofaction{Z}{v_2}$, and
    from the definition of outgoing edges of $\pinode$, we get $\pinode\cgedge{SZ} v_2$ in $\cgraph{e_n,B,\vec C}$.
    From the third item and the proof of  Lemma~\ref{lem:causal}, we get either $v_2\cgedgeD{ZS} v$ in $\cgraph{e_n,B}$ or
    $(v_2,Z)=(v,S)$. All together, we get $\pinode\cgedgeD{SS}v$ in $\cgraph{e_n,B,\vec C}$.
    By definition of  ${C_S^{\pi}}'$ (first line), it contains $\procofactionv{X}{v}$ .
  \item $\vstart$ and $v$ in $e_1\cdots e_{n-1}$.
    If $\vstart\cgedgeD{SS} v$ in $\cgraph{e_1\cdots e_{n-1}}$, then $\procofaction{X}{v}\in C_R^{\pi}$ holds immediately by induction hypothesis.
    Otherwise, there are $v_1,v_2,v_3,v_4,Y,Z,q$ such that
    \begin{itemize}
    \item either $\vstart \cgedgeD{SY} v_1$ in $\cgraph{e_1\cdots e_{n-1}}$, or $(v,S)=(v_1,Y)$
    \item $v_1\cgedge{YZ}v_2$
    \item either $v_2\cgedgeD{ZS} v_3$ in $\cgraph{e_1\cdots e_n}$, with both $v_2$ and $v_3$ in $e_n$, or $(v_2,Z)=(v_3,S)$
    \item $v_3\cgedgeD{SS}v_4$ due to rule 4 in Fig.~\ref{fig:rules-cgedgeD}, i.e., $v_3$ is a matched send to $q$ and $v_4$ is an unmatched send to $q$ 
    \item either $v_4 \cgedgeD{SS} v$ in $\cgraph{e_1\cdots e_{n-1}}$, or $(v_4,T)=(v,S)$
    \end{itemize}
    From the first item, by induction hypothesis, we get
    $\procofaction{Y}{v_1}\in C_X^{\pi}$.
    From the second item, we get $\procofaction{Y}{v_1}=\procofaction{Z}{v_2}$, and
    from the definition of outgoing edges of $\pinode$, we get $\pinode\cgedge{SZ} v_2$ in $\cgraph{e_n,B,\vec C}$.
    From the third item and the proof of  Lemma~\ref{lem:causal}, we get either $v_2\cgedgeD{ZS} v_3$ in $\cgraph{e_n,B}$ or
    $(v_2,Z)=(v_3,S)$.
    From the fourth item, we get $v_3\cgedgeD{SS}\lambdanode_q$ in $\cgraph{e_n,B}$.
    To sum up, we have $\pinode\cgedge{SS}\lambdanode_q$ in $\cgraph{e_n,B,\vec C}$.
    By definition of $C_X^{\pi}$, we therefore have $C_{X,q}\subseteq C_X^{\pi}$.
    From the fifth item, we get by the proof of  Lemma~\ref{lem:causal} that $\procofaction{X}{v}\in C_{X,q}$, which ends this case.
  \end{itemize}

\item Finally, let us finish the proof of the converse implication, and
  show the remaining case, i.e., let us show that $\procofactionv{S}{\vstart}\in C_S^{\pi}$.
  This is immediate from the definition of ${C_S^{\pi}}'$ (cfr. the set
  $\{\procofactionv{S}{v}\mid\procofactionv{R}{v}=\pi\mbox{ \& } X=S\}$).
\end{itemize}

We are done with proving that $C_X^{\pi}=\{\procofactionv{X}{v}\mid (v\cap X\neq\emptyset \mbox{ \& } 
\vstart\cgedgeD{SS}v\mbox{ in }\cgraph{e'}) \mbox{ or } (v,X)=(\vstart,S)\}$.
It is now time to conclude with the proof of Lemma~\ref{lem:feasible-is-regular} itself.

Let $e'$ and $e'' =e_1\cdots e_n\cdot \send{\pi}{q}{\amessage}\cdot \rec{\pi}{q}{\amessage}$ with
$e_1,\cdots e_n\in S^{\leq k}R^{\leq k}$ be fixed such that $msc(e') = msc(e'')$.

$\Leftarrow$ Let us assume that $e'$ is a \kSable{k} feasible execution of $\system'$ and let us show that
$$
(\globalstate{l_0}, B_0, (\emptyset,\emptyset),\bot)
\transitionCD{e_1}{k}\dots\transitionCD{e_n}{k}
(\globalstate{l'}, B', \vec{C'}, \destpi).
$$
for some $B',\vec C',\destpi$ with $\pi\not\in C_{R,\destpi}$.
By definition of $\transitionCD{e}{k}$, $B'$, $\vec C'$ and $\destpi$ are uniquely determined, and
it is enough to prove that $\destpi\not\in C_R^{\pi}$. Let us assume by absurd that
$\destpi\in C_R^{\pi}$. Then, by the property we just proved, there is $v$ such that 
$\procofactionv{R}{v}=\destpi$, $v\cap R\neq\emptyset$, and
$\vstart\cgedgeD{SS}v$ in $\cgraph{e_1'\cdots e_n'}$. So we get $\vstart\cgedgeD{SS}v\cgedge{RR}\vstop$ in
$\cgraph{e'}$, and by Lemma~\ref{lem:graph-carac-feas}, $e'$ should not be feasible: contradiction.
Finally, $\pi\not\in C_{R,\pi}$ because $e'$, as an execution of $\system'$, $msc(e')$ satisfies causal delivery.

$\Rightarrow$ Let us assume that
$$
(\globalstate{l_0}, B_0, (\emptyset,\emptyset),\bot)
\transitionCD{e_1}{k}\dots\transitionCD{e_n}{k}
(\globalstate{l'}, B', \vec{C'}, \destpi).
$$
for some $B',\vec C',\destpi$ with $\pi\not\in C_{R,\destpi}$, and let us show that $e'$ is a \kSable{k} feasible execution of $\system'$.
From the definition of $\transitionCD{e}{k}$, we get
$$
(\globalstate{l_0}, B_0)
\transitionKE{e_1}{k}\dots\transitionKE{e_n}{k}
(\globalstate{l'}, B')
$$
and from Lemma~\ref{lem:causal}, $msc(e')$ is \kSous{k}. 
Since the last two actions

\noindent $\send{\pi}{q}{\amessage}~\cdot~\rec{\pi}{q}{\amessage}$ can be placed in a new \kE{k}, and since they
do not break causal delivery (because $\pi\not\in C_{R,\destpi}$), $e'$ is a \kSable{k} execution of $\system'$.
It remains to show that $e'$ is feasible. Again, let us reason by contradiction and assume that $e'$ is not
feasible. By Lemma~\ref{lem:graph-carac-feas}, there is $v$ such that $\vstart\cgedgeD{SS}v\cgedge{RR}\vstop$ in
$\cgraph{e'}$. In other words, 
$\procofactionv{R}{v}=\destpi$, $v\cap R\neq\emptyset$, and
$\vstart\cgedgeD{SS}v$ in 
$\cgraph{e'}$. 
So, by the property we just proved, $\destpi\in C_{R}^{\pi}$,
Hence the contradiction. \qed
\end{proof}

\end{proof}
\fi

\begin{comp}
In \cite{DBLP:conf/cav/BouajjaniEJQ18} the authors verify that an execution is feasible with a \emph{monitor} which reviews the actions of the execution and adds processes that no longer are allowed to send a message to the receiver of $\pi$. Unfortunately, we have here a similar problem that the one mentioned in the previous comparison paragraph. According to their monitor, the following execution $e' = \deviate{e \cdot r} $ (see its MSC in Fig.~\ref{fig:problematic-exe}a in Appendix~\ref{sec:additional}) is feasible, i.e., is runnable in $\system'$ and $e \cdot r$ is runnable in $\system$. 
\begin{align*}
e' =~  & \send{q}{\pi}{(r,\amessage_1)}\cdot\rec{q}{\pi}{(r,\amessage_1)}\cdot
\send{q}{s}{\amessage_2}\cdot\rec{q}{s}{\amessage_2}
\cdot \\ 
& \send{p}{s}{\amessage_3}\cdot\rec{p}{s}{\amessage_3}\cdot \send{p}{r}{\amessage_4}\cdot\rec{p}{r}{\amessage_4} \cdot  \\
& \send{\pi}{r}{\amessage_1}\cdot\rec{\pi}{r}{\amessage_4}
\end{align*}

However,  this execution is not feasible because there is a causal dependency between $\amessage_1$ and $\amessage_3$. 
In \cite{DBLP:conf/cav/BouajjaniEJQ18} this execution would then be considered as feasible and therefore would belong to set $sTr_k(\system')$. Yet there is no corresponding execution in $asTr(\system)$, the comparison and therefore the \kSity{k}, could be distorted and appear as a false negative.  
\end{comp}

\paragraph{\bf Recognition of bad executions.}
Finally, we define a non-deterministic finite state
automaton that recognizes MSCs of bad executions, i.e., feasible executions $e'=\deviate{e\cdot r}$ such that $e\cdot r$ is not \kSable{k}. 
We come back to the ``non-extended'' conflict graph, without edges of the form $\cgedgeD{XY}$.
Let $\succs{v}=\{v'\in V\mid v\to^* v'\}$ be the set of vertices
reachable from $v$, 
and let
$\preds{v}=\{v'\in V\mid v'\to^* v\}$ be the set of vertices co-reachable
from $v$. For a set of vertices
$U\subseteq V$, let $\succs{U}=\bigcup\{\succs{v}\mid v\in U\}$,
and $\preds{U}=\bigcup\{\preds{v}\mid v\in U\}$.

\begin{lemma}\label{lem:bad-characterization}
The feasible execution $e'$ is bad iff  one of the two
holds
\begin{enumerate}
\item  $\vstart\cgedge{}^*\cgedge{RS}\cgedge{}^*\vstop$, or 
\item the size of the set $\succs{\vstart}\cap\preds{\vstop}$ is greater or equal to $k+2$.
\end{enumerate}
\end{lemma}

\iflong
\begin{proof}
  \begin{quote}
{\bf Lemma 5}
The feasible execution $e'$ is bad iff  one of the two
holds
\begin{enumerate}
\item 
 $\vstart\cgedge{}^*\cgedge{RS}\cgedge{}^*\vstop$, or 
\item 
the size of the set $\succs{\vstart}\cap\preds{\vstop}$ is greater or equal to $k+2$.
\end{enumerate}
\end{quote}

\begin{proof}

Since $msc(e')$ is \kSous{k} and $e' = \deviate{e\cdot r}$, $msc(e)$ (without the last reception $r$)
is \kSous{k}.
By Theorem~\ref{thm:k-sity-scc}, $e'$ is bad if and only if
$\cgraph{e\cdot r}$ contains either a cyclic path
with an RS edge, or a SCC with of size $\geq k+1$.
This cyclic path (resp. SCC) must contain
the vertex associated with the last receive $r$ of $e\cdot r$. In $\cgraph{e'}$, this cyclic (resp. SCC) corresponds to a path
from $\vstart$ to $\vstop$ (resp. the set of vertices that
are both reachable from $\vstart$ and co-reachable from $\vstop$). Since
the $\vstart$ and $\vstop$ account for the same node in the conflict graph
of $e\cdot r$, the size of the SCC is one less than the size of the set
$\succs{\vstart}\cap\preds{\vstop}$. \qed
\end{proof}
\end{proof}
\fi

In order to determine whether a given message exchange
$v$ of $\cgraph{e'}$ should
be counted as reachable (resp. co-reachable),
we will compute at the entry and exit of every \kE{k} of $e'$
which processes are ``reachable'' or ``co-reachable''.

\begin{example}\label{ex:reach-co-reach}
(Reachable and co-reachable processes)\\
\begin{minipage}[b]{8.5cm}
\vspace*{0.2cm}
Consider the MSC  on the right composed of five \kE{1}s.
While sending
message $(s,\amessage_0)$ that corresponds to $\vstart$,
process $r$ becomes ``reachable'': any subsequent message exchange that
involves $r$ corresponds to a vertex of the conflict graph
that is reachable from $\vstart$. While sending $\amessage_2$, process $s$ becomes ``reachable'',
because process $r$ will be reachable when it will receive message
$\amessage_2$.
Similary, $q$ becomes reachable after recei\-ving $\amessage_3$
because $r$ was reachable when it sent $\amessage_3$, and $p$ becomes
reachable after receiving  $\amessage_4$ because $q$ was rea-

\end{minipage}
   \hfill
\begin{minipage}[t]{4cm}
\begin{center}
\begin{tikzpicture}[scale = 0.8]
\draw (2,-4.3) node[above] {$msc(e)$};
  \draw (0,0.3) node{$p$} ;
  \draw (1,0.3) node{$q$} ;
    \draw (2,0.3) node{$r$} ;
  \draw (3,0.3) node{$s$} ;
    \draw (4,0.3) node{$\pi$} ;

\draw (0,0) edge (0,-3.5); 
\draw (1,0) edge (1,-3.5);
\draw (2,0) edge (2,-3.5); 
\draw (3,0) edge (3,-3.5); 
\draw (4,0) edge (4,-3.5);
\draw[>=latex, ->] (2,-0.5) to node [pos=0.25,above, scale = 0.85] {$(s,\amessage_0)$} (4,-0.5);
\draw[>=latex, ->] (0,-1) to node [above, sloped, pos = 0.2, scale = 0.9] {$\amessage_1$} (3,-1);
\draw[>=latex, ->] (3,-1.5) to node [above, scale = 0.9] {$\amessage_2$} (2,-1.5);
\draw[>=latex, ->] (2,-2) to node [above, scale = 0.9] {$\amessage_3$} (1,-2);
\draw[>=latex, ->] (1,-2.5) to node [above, scale = 0.9] {$\amessage_4$} (0,-2.5);
\draw[>=latex, ->] (4,-3) to node [above, scale = 0.9] {$\amessage_0$} (3,-3);
\end{tikzpicture}
\end{center}
\end{minipage}
\\ chable  when it sent
it. 
Co-reachability works similarly, but reasoning backwards on the timelines.
For instance, process $s$ stops being ``co-reachable'' while it receives
$\amessage_0$, process $r$ stops being co-reachable after it receives $\amessage_2$,
and process $p$ stops being co-reachable by sending $\amessage_1$.
The only message that is sent by a process being both reachable
and co-reachable at the instant of the sending is $\amessage_2$, therefore
it is the only message that will be counted as contributing to the SCC.
\end{example}

More formally, let
$e$ be sequence of actions, 
$\cgraph{e}$  its conflict graph and $P, Q$ two sets of processes, 
$
\succse{e}{P}=\mathsf{Post}^*\Big(\{v\mid \procs{v}\cap P\neq\emptyset\}\Big)$
 and 
$\predse{e}{Q}=\mathsf{Pre}^*\Big(\{v\mid \procs{v}\cap Q\neq\emptyset\}\Big)
$
are introduced to represent 
the local view through  \kE{k}s of $\succs{\vstart}$ and $\preds{\vstop}$.
For instance, for $e$ as in Example~\ref{ex:reach-co-reach},
we get $\succse{e}{\{\pi\}} = \{(s,\amessage_0), \amessage_2, \amessage_3,\amessage_4,\amessage_0\}$ and 
$\predse{e}{\{\pi\}} = \{\amessage_0,\amessage_2,\amessage_1,(s,\amessage_0)\}$.
In each \kE{k} $e_i$ the size of the intersection between $\succse{e_i}{P}$ and $\predse{e_i}{Q}$ will give the local contribution of the current \kE{k} to the calculation of the size of the global SCC. In the transition relation $\transitionBV{e}{k}$ this value is stored in variable $\cont$.
The last ingredient to consider is to recognise if an edge RS belongs to the SCC.
To this aim, we use a function $\lastisRec: \procSet \to \{ \cttrue, \ctfalse \}$ that for each process stores the information whether the last action in the previous \kE{k} was a reception or not.  Then depending on the value of this variable and if a node is in the  current SCC or not the value of $\sawRS$ is set accordingly.

\begin{figure}[t]
 \begin{prooftree}
\AxiomC {\small
	\stackanchor{ \stackanchor{$P'=\procs{\succse{e}{P}} \qquad Q=\procs{\predse{e}{Q'}}$}{ $SCC_e = \succse{e}{P}\cap\predse{e}{Q' }$}}
		{ \stackanchor {$\cont' = \mathsf{min}(k+2,\cont + n) \quad \mbox{where } n= | SCC_e |$}
			{\stackanchor {\stackanchor{$
\lastisRec'(q) \Leftrightarrow (\exists v\in SCC_e. \procofaction{R}{v}=q\wedge v\cap R\neq\emptyset) \vee $}{$(\lastisRec(q) \wedge\! \not \exists v \in V.
\procofaction{S}{v}=q )
$}}
				{\stackanchor { \stackanchor{$
\sawRS' = \sawRS \vee $}{$(\exists v\in SCC_e)(\exists p\in \procSet\setminus\{\pi\})\
\procofactionv{S}{v} = p \wedge \lastisRec(p) \wedge p\in P \cap Q 
$}}}}}}
\UnaryInfC{ \small $
(P,Q, \cont, \sawRS, \lastisRec)
\transitionBV{e}{k}
(P', Q', \cont', \sawRS', \lastisRec')
$}
\end{prooftree}
\vspace*{-0.5cm}
\caption{Definition of the relation $\transitionBV{e}{k}$
}
 \label{fig:BVrel}
 \end{figure}


The transition relation $\transitionBV{e}{k}$  defined in Fig.~\ref{fig:BVrel} deals with abstract confi\-gurations of the form
$(P,Q,\cont,\sawRS, {\lastisRec'})$ 
where $P,Q\subseteq \procSet$, $\sawRS$ is a boolean value,
and $\cont$ is a counter bounded by $k+2$. We denote by $\lastisRec_0$ the function where all $\lastisRec(p)=\ctfalse$ for all $p\in \procSet$.

\begin{lemma}\label{lem:transitionBV}
Let $e'$ be a feasible \kSable{k} execution of $\system'$. Then $e'$ is a bad execution iff there are $e'' =e_1\cdots e_n \cdot \send{\pi}{q}{\amessage}\cdot \rec{\pi}{q}{\amessage}$ with
  $e_1,\ldots ,e_n\in S^{\leq k}R^{\leq k}$ and $msc(e')=msc(e'')$, $P',Q\subseteq \procSet$, $\sawRS\in\{\cttrue,\ctfalse\}$, $\cont\in\{0,\dots,k+2\}$,
such that
$$
(\{\pi\},Q,0,\ctfalse,\lastisRec_0)
\transitionBV{e_1}{k}\dots\transitionBV{e_n}{k}
(P',\{\pi\},\cont,\sawRS,\lastisRec)
$$
and at least one of the two holds: either $\sawRS=\cttrue$, or $\cont=k+2$.
\end{lemma}
\iflong
\begin{proof}
  
 
\begin{quote}
{\bf Lemma 6}
	Let $e'$ a feasible \kSable{k} execution of $\system'$. Then $e'$ is a bad execution iff there are $e'' =e_1\cdots e_n \cdot \send{\pi}{q}{\amessage}\cdot \rec{\pi}{q}{\amessage}$ with
  $e_1,\ldots ,e_n\in S^{\leq k}R^{\leq k}$ and $msc(e')=msc(e'')$, $P',Q\subseteq \procSet$, $\sawRS\in\{\cttrue,\ctfalse\}$, $\cont\in\{0,\dots,k+2\}$,
such that
$$
(\{\pi\},Q,0,\ctfalse,\lastisRec_0)
\transitionBV{e_1}{k}\dots\transitionBV{e_n}{k}
(P',\{\pi\},\cont,\sawRS,\lastisRec)
$$
and at least one of the two holds: either $\sawRS=\cttrue$, or $\cont=k+2$.
%
%
\end{quote}

 \begin{proof}
 
 $\Rightarrow$ 
Let us suppose $e'$ be a \kSable{k} bad and feasible execution such that $e'' =e_1\cdots e_n\cdot \send{\pi}{q}{\amessage}\cdot \rec{\pi}{q}{\amessage}$ with $msc(e')=msc(e'')$.
We show that $$
(\{\pi\},Q,\ctfalse,0)
\transitionBV{e_1}{k}\dots\transitionBV{e_n}{k}
(P',\{\pi\},\sawRS,\cont)
$$
for some $P'$,   $Q$
and 
either $\sawRS=\cttrue$, or $\cont=k+2$.

We proceed by induction on $n$.
\begin{description}
\item[Base n=2]
 Notice that,  for a feasible execution, there are at least two \kE{k}s as the deviation cannot fit a single \kE{k}: the send from process $\pi$ to the original recipient must follow the reception of the deviated message, thus it has to belong to a subsequent \kE{k}.
 Then $e'' = e_1 \cdot e_2$ and 
 we show $(\{\pi\},Q,\ctfalse,0)
\transitionBV{e_1}{k} (P',Q',\sawRS',\cont)  \transitionBV{e_2}{k}
(P'',\{\pi\},\sawRS,\cont')$.
 
 By Lemma \ref{lem:bad-characterization}, we have that either $\vstart\cgedge{}^*\cgedge{RS}\cgedge{}^*\vstop$, or 
the size of the set $\succs{\vstart}\cap\preds{\vstop}$ is greater or equal to $k+2$.

If $\vstart\cgedge{}^*\cgedge{RS}\cgedge{}^*\vstop$, then since a label RS cannot exists in a local conflict graph,
there exist two paths $\vstart\cgedge{}^*v_1$ in $\cgraph{e_1}$ and $v_2\cgedge{}^*\vstop$ in $\cgraph{e_2}$, with $\procofactionv{R}{v_1}=\procofactionv{S}{v_2}$.
We have that $v_2 \in \predse{e_2}{\pi}$ and $\lastisRec(\procofactionv{S}{v_2})$ is $\cttrue$, thus $\sawRS$ becomes $\cttrue$ concluding this part of the proof.

Now suppose that  the size of the set $\succs{\vstart}\cap\preds{\vstop}$ is greater or equal to $k+2$. We show that all nodes in $\succs{\vstart}\cap\preds{\vstop}$ have been counted either in the first or in the second \kE{k}.
Take $v \in \succs{\vstart}\cap\preds{\vstop}$ then there exists a path $ \vstart \cgedge{}^*v \cgedge{}^*\vstop$ and 
 $v$ is an exchange that belongs either to the first or the second \kE{k}. 
If $v$ belongs to the first one then we can divide previous path in two parts such that $ \vstart \cgedge{}^*v \cgedge{}^*v_1$ is in $\cgraph{e_1}$,  $ v_2 \cgedge{}^*\vstop$
is in $\cgraph{e_2}$ and $\procs{v_1} \cup \procs{v_2} = \{p\} \neq \emptyset$. 
From this it follows that process $p \in Q'$ and thus $v \in \predse{e_1}{Q'}$. Moreover, $v \in \succse{e_1}{\pi}$ and therefore the node $v$ is counted in the first \kE{k}.

Similarly, if $v$ belongs to the second \kE{k}, we can divide the path into two parts such that $ \vstart \cgedge{}^*v_1$ is in $\cgraph{e_1}$,  $ v_2 \cgedge{}^*v \cgedge{}^*\vstop$
is in $\cgraph{e_2}$ and $\procs{v_1} \cup \procs{v_2} = \{p\} \neq \emptyset$.
From this it follows that process $p \in P'$ and thus $v \in \succse{e_2}{P'}$. Moreover, $v \in \predse{e_2}{\pi}$ and therefore the node $v$ is counted in the second \kE{k}.

Thus all nodes in $\succs{\vstart}\cap\preds{\vstop}$ are considered and if $\succs{\vstart}\cap\preds{\vstop} \geq k+2$ so is variable $\cont'$, concluding this part of the proof.

\item[Inductive step]
It is an easy generalisation of what has been said in the previous part of the proof. By considering that by inductive hypothesis sets $\succse{e_i}{P}$ and $\predse{e_i}{Q} $ contains respectively all the processes that are reachable from the exchange to process $\pi$ and are co-reachable from the exchange from process $\pi$.
\end{description}

$\Leftarrow$
Let $e'$ a \kSable{k} feasible execution,  $e'' =e_1\cdots e_n\cdot \send{\pi}{q}{\amessage}\cdot \rec{\pi}{q}{\amessage}$ with $msc(e')=msc(e'')$, and 
$P'$, $Q \subseteq \procSet$, $\sawRS\in\{\cttrue,\ctfalse\}$, $\cont\in\{0,\dots,k+2\}$ such that
$$
(\{\pi\},Q,\ctfalse,0)
\transitionBV{e_1}{k}\dots\transitionBV{e_n}{k}
(P',\{\pi\},\sawRS,\cont)
$$
We have either $\sawRS=\cttrue$ or $\cont=k+2$. 

\begin{enumerate}

\item We suppose that $\sawRS=\cttrue$. 
If $\sawRS=\cttrue$ then $\exists e_i$ where $\sawRS = \ctfalse$ and $\sawRS' = \cttrue$. In this \kE{k}, $\exists p \in \procSet$ such that $p \in P$, $\lastisRec(p) = \cttrue$ and $\exists v$ such that $\procofactionv{S}{v}= p$ and $v \in \predse{e'}{Q'}$. 
Since $p \in P$, then there is a path $\vstart \cgedge{}^* \cgedge{RS} v$ in $\cgraph{e'}$. On the other hand, since $v \in \predse{e'}{Q'}$ then $v \in \preds{\vstop}$ and there is a path $v \cgedge{}^* \vstop$ in $\cgraph{e'}$.
Therefore, there is a path $\vstart \cgedge{}^* \cgedge{RS} v \cgedge{}^* \vstop$ in $\cgraph{e'}$ and so $e'$ is bad. 

\item We suppose that $\cont=k+2$. As previously, $e'$ is feasible by Lemma~\ref{lem:feasible-is-regular}.
Each $v$ belongs to $\succse{e_i}{P_i} \cap \predse{e_i}{Q'_i} \setminus \vstart$ also belongs to $\succs{\vstart} \cap \preds{\vstop}$ then $\mid \succs{\vstart} \cap \preds{\vstop}  \mid \geq k+2$. Therefore, $e'$ is bad. 

\end{enumerate}

Therefore, in  both cases, $e'$ is feasible and bad, concluding the proof. 
\qed
 \end{proof}

\end{proof}
\fi

\begin{comp} 
As for the notion of feasibility, to determine if an execution is bad,  in \cite{DBLP:conf/cav/BouajjaniEJQ18} the authors use a monitor that builds a path between the send to process $\pi$ and the send from $\pi$. 
In addition to the problems related to the wrong characterisation of \kSity{k}, this monitor not only
can detect  an $RS$ edge  when there should be none, but also it can miss them when they exist. In general, the problem arises because the path is constructed by considering only an endpoint at the time (see Example~\ref{ex:counter-bad-def} in Appendix~\ref{sec:additional_comparison} for more explanations). 
\end{comp}
%
%

We can finally conclude that:
\begin{theorem}\label{thm:ksity-decidability}
The \kSity{k} of a system $\system$ is decidable for  $k\geq 1 $. 
\end{theorem}

\iflong
\begin{proof}
  \begin{quote}
{\bf Theorem 4}
The \kSity{k} of a system $\system$ is decidable for $k \geq 1$. 
\end{quote} 
 
 \begin{proof}
  Let $\system$ be fixed.
  By Lemmata~\ref{lem:bad-is-bad}, \ref{lem:feasible-is-regular}, and \ref{lem:transitionBV},
  $\system$ is not \kSable{k} if and only if there is a sequence of actions
  $e'=e_1'\cdots e_{n}'\cdot s\cdot{}r$
  such that $e_i\in S^{\leq k}R^{\leq k}$,   $s=\send{\pi}{q}{\amessage}$, $r=\rec{\pi}{q}{\amessage}$,
  $$
  (\globalstate{l_0}, B_0, (\emptyset,\emptyset), \bot)\transitionCD{e_1'}{k}\dots\transitionCD{e_{n}'}{k}
  (\globalstate{l'}, B', \vec C',q)
  $$
  and
  $$
  (\{\pi\},Q,\ctfalse,0)
  \transitionBV{e_1'}{k}\dots\transitionBV{e_{n}'}{k}\transitionBV{s\cdot{}r}{k}
  (P',\{\pi\},\sawRS,\cont)
  $$
  for some $\globalstate{l'},B',\vec C',Q,P'$ with $\pi\not\in C_{R,q}$.
  Since both relations $\transitionCD{e}{k}$ and $\transitionBV{e}{k}$ are finite state, the existence of such a sequence of actions is decidable. \qed
 \end{proof}

\end{proof}
\fi


\section{\kSity{k} for peer-to-peer systems} \label{section:p2p}

In this section, we will apply \kSity{k} to peer-to-peer systems. A peer-to-peer system is a composition of communicating automata where each  pair of machines exchange messages via two private FIFO buffers, one per direction of communication. 
Precise formal definitions, lemmata and theorems can be found in  Appendix~\ref{sec:additional_p2p}. Here we only give a quick insight on what changes w.r.t. the mailbox setting.


Causal delivery reveals the order imposed by FIFO buffers. Definition \ref{def:causal-delivery}
 must then be  adapted to account for peer-to-peer communication. For instance, two messages that are sent to a same process $p$ by two different processes can be received by $p$ in any order, regardless of any causal dependency between the two sends. Thus, checking causal delivery in peer-to-peer systems is easier than in the mailbox setting, as we do not have to carry information on causal dependencies.  

Within a peer-to-peer architecture, MSCs and conflict graphs are defined as within a mailbox communication. Indeed, they represents dependencies over machines, i.e., the order in which the actions can be done on a given machine, and over the send and the reception of a same message, and they do not depend on   the type of communication. The notion of \kE{k} remains also unchanged.

\paragraph{\bf Decidability of reachability for \kSable{k} peer-to-peer systems.}

To establish the decidability of reachability for \kSable{k} peer-to-peer systems, 
we define a transition relation $\ptransitionKE{e}{k}$ (see Fig.~\ref{fig:transitionCDp2p} in Appendix~\ref{sec:additional_p2p}) for a sequence of action $e$ describing a \kE{k}. 
As for mailbox systems,  if a send action is unmatched  in the current \kE{k}, it will stay orphan forever.  Moreover, 
after a process $p$ sent an orphan message to a process $q$, $p$ is forbidden to send any matched message to $q$. 
Nonetheless, as a consequence of the simpler definition of causal delivery,  , we no longer need to work on the conflict graph. Summary nodes and extended edges are not needed and all the necessary information is in  function $B$ that solely  contains all the forbidden senders for process $p$. 

The characterisation of a \kSable{k} execution is the same as for  mailbox systems as the type of communication is not relevant.
%
We can thus conclude, as within mailbox communication, that  reachability is decidable. 

%
\begin{theorem}\label{thm:reachability-dec-p2p}
Let $\system$ be a \kSable{k} system and $\globalstate{l}$ a global control state of $\system$. The problem whether there exists $e \in asEx(\system)$ and $\B$ such that $(\globalstate{l_0}, \B_0) \xRightarrow{e} (\globalstate{l}, \B)$ is decidable. 
\end{theorem}
\paragraph{\bf Decidability of \kSity{k} for peer-to-peer systems.}

As in mailbox system, the detection of a borderline execution determines whether a system is \kSable{k}.
%

The relation transition $\ptransitionCD{e}{k}$ allowing to obtain feasible executions can be found in Fig.~\ref{fig:transitionfeasp2p} in Appendix~\ref{sec:additional_p2p}.
%
%
Differently from the mailbox setting, we  need to save not only the recipient $\finddest$ but also the sender of the delayed message (information stored in variable $\findexp$). The transition rule then  checks that there is no message that is violating causal delivery, \ie  there is no message sent by $\findexp$ to $\finddest$ after the deviation.
%
%
Finally the recognition of bad execution, works in the same way as for mailbox systems.
The characterisation of a bad execution and the definition of $\ptransitionBV{e}{k}$ are, therefore, the same.

As for mailbox systems, we can, thus, conclude that for a given $k$,  \kSity{k} is decidable.

\begin{theorem}\label{thm:ksity-decidability-p2p}
The \kSity{k} of a system $\system$ is decidable for  $k\geq 1 $. 
\end{theorem}



\section{Concluding remarks and related works}
\label{sec:conc}

In this paper we have studied  \kSity{k} for mailbox and peer-to-peer systems. 
We have corrected the reachability and decidability  proofs 
  given in~\cite{DBLP:conf/cav/BouajjaniEJQ18}. The flaws in \cite{DBLP:conf/cav/BouajjaniEJQ18}
 concern fundamental points and  we had to propose a considerably different approach. 
The extended edges of the conflict graph, and the
graph-theoretic characterisation of causal delivery as well as summary nodes, have no equivalent in
\cite{DBLP:conf/cav/BouajjaniEJQ18}.
Transition relations $\transitionCD{e}{k}$ and $\transitionBV{e}{k}$ building on the graph-theoretic characterisations of causal delivery and \kSity{k},  depart considerably from  the proposal in \cite{DBLP:conf/cav/BouajjaniEJQ18}.

We conclude by commenting on some other related works. 
The idea of ``communication layers''  
is present in the early works of
Elrad and Francez~\cite{ElradF82} or Chou and Gafni~\cite{ChouG88}.
More recently, Chaouch-Saad~\emph{et al.}~\cite{Chaouch-SaadCM09}
verified some consensus algorithms using the Heard-Of Model
that proceeds by ``communication-closed rounds''.
The concept that an asynchronous system may have 
an ``equivalent'' synchronous
counterpart has also been widely studied.
Lipton's reduction~\cite{Lipton75} res\-chedules an execution
so as to move the receive actions as close as possible from their corresponding
send. Reduction recently received
an increasing interest for verification purpose, e.g.
by Kragl~\emph{et al.}~\cite{KraglQH18}, or
Gleissenthal~\emph{et al.}~\cite{GleissenthallKB19}.

Existentially bounded communication systems have been studied
by Ge\-nest~\emph{et al.}~\cite{DBLP:journals/fuin/GenestKM07,DBLP:conf/lata/Muscholl10}: a system is existentially $k$-bounded if any execution can be rescheduled
in order to become $k$-bounded. This approach targets a broader class of systems than
$k$-synchronizability, because it does not require that the execution
can be chopped in communication-closed rounds. In the perspective of the current
work, an interesting result is the decidability of
existential $k$-boundedness for deadlock-free  systems of communicating machines
with peer-to-peer channels. Despite the more general definition,
these older results are incomparable with the present ones,
that deal with systems communicating with mailboxes, and
not peer-to-peer channels. 



Basu and Bultan studied a notion they also called synchronizability,
but it differs from the notion studied in the present work; synchronizability
and \kSity{k} define incomparable
classes of communicating systems. The proofs of the decidability
of synchronizability~\cite{BasuBO12,BasuB16} were shown to have
flaws by Finkel and Lozes~\cite{DBLP:conf/icalp/FinkelL17}. A question
left open in their paper is whether synchroni\-zability  is decidable
for mailbox communications, as originally claimed by Basu and Bultan. 
Akroun and Sala\"un defined also a property they called
stability~\cite{DBLP:journals/fmsd/AkrounS18} and that shares many similarities
with the synchronizability notion in \cite{BasuB16}.

Context-bounded model-checking is yet another approach for the automatic
verification of concurrent systems. La Torre~\emph{et al.} studied systems of
communicating machines extended with a calling stack, and showed that
under some conditions on the interplay between stack
actions and communications, context-bounded reachability
was decidable~\cite{DBLP:conf/tacas/TorreMP08}.
A context-switch is found in an execution each time two consecutive actions are performed by a different participant. Thus,  while \kSity{k} limits the number of consecutive sendings, bounded context-switch analysis limits the number 
of times two consecutive actions are performed by two different processes.

As for future work, it would be interesting
to explore how both context-boundedness and communication-closed rounds
could be composed. Moreover refinements of the definition of \kSity{k}
 can also be considered. For instance, we conjecture that the current development can be greatly simplified if we forbid linearisation that do not correspond to actual executions.



\newpage

\bibliographystyle{splncs04}
 \bibliography{biblio}

\newpage

\newpage
\appendix


\section{Comparison with \cite{DBLP:conf/cav/BouajjaniEJQ18} \\(Examples and additional material)}\label{sec:additional_comparison}

Let $po$ and $src$ be the partial orders on the set of actions obtained respectively from  $\prec_{po}$ and $\prec_{src}$ by assuming that  if $i \prec_{po} j$ then $\lambda(i) < \lambda(j) \in po$ and if $i \prec_{src} j$ then $\lambda(i) < \lambda(j) \in src$.


\begin{example}\label{ex:counter-bad-def}[Problems with the development in \cite{DBLP:conf/cav/BouajjaniEJQ18}]
Fig.~\ref{fig:problematic-exe}b depicts the MSC associated with a feasible execution feasible that does not contain label $RS$. The monitor in \cite{DBLP:conf/cav/BouajjaniEJQ18} considers the reception of $\amessage_2$ followed by the send of $\amessage_3$. A label $RS$ is thus wrongly detected. 

Fig.~\ref{fig:problematic-exe}c, instead, depicts the MSC associated with an execution feasible but bad. With the monitor  in \cite{DBLP:conf/cav/BouajjaniEJQ18}, the action seen after the send of $\amessage_3$ is the send of $\amessage_4$ and so the existing label $RS$ is ignored at the profit of a non existing  label $SS$.  
\end{example}

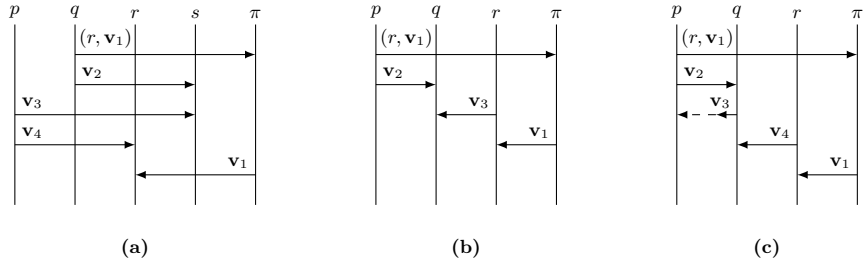
\begin{figure}[t]

\begin{center}

\begin{tikzpicture}[scale = 0.8,every node/.style={scale=0.8}]

\begin{scope}
	\coordinate(Aa) at (0,0) ;
	\coordinate (Ab) at (0,-3) ;
	\coordinate (Ba) at (1,0) ;
	\coordinate (Bb) at (1,-3) ;
	\coordinate (Ca) at (2,0) ;
	\coordinate (Cb) at (2,-3) ;
	\coordinate (Pa) at (3,0) ;
	\coordinate (Pb) at (3,-3) ;
	\coordinate (Da) at (4,0) ;
	\coordinate (Db) at (4,-3) ;
	\draw (Aa) node[above]{$p$} ;
	\draw (Ba) node[above]{$q$} ;
	\draw (Ca) node[above]{$r$} ;
	\draw (Pa) node[above]{$s$} ;
	\draw (Da) node[above]{$\pi$} ;
	\draw (Aa) -- (Ab) ;
	\draw (Ba) -- (Bb) ;
	\draw (Ca) -- (Cb) ;
	\draw (Pa) -- (Pb) ;
	\draw (Da) -- (Db) ;
	\coordinate (sa) at (1,-0.5); 
	\coordinate (ra) at (4,-0.5); 
	\draw[>=latex,->] (sa) to node[above, sloped, pos = 0.17]{$(r,\amessage_1)$}(ra);
	
	\coordinate (sb) at (1, -1); 
	\coordinate (rb) at (3, -1); 
	\draw[>=latex,->] (sb) -- (rb); 
	\draw (sb) node[above right]{$\amessage_2$} ;
	
	\coordinate (sc) at (0, -1.5); 
	\coordinate (rc) at (3, -1.5); 
	\draw[>=latex,->] (sc) -- (rc); 
	\draw (sc) node[above right]{$\amessage_3$};

	\coordinate (sd) at (0, -2); 
	\coordinate (rd) at (2, -2); 
	\draw[>=latex,->] (sd) -- (rd); 
	\draw (sd) node[above right]{$\amessage_4$};
	
	\coordinate (sd) at (4, -2.5); 
	\coordinate (rd) at (2, -2.5); 
	\draw[>=latex,->] (sd) -- (rd); 
	\draw (sd) node[above left]{$\amessage_1$};
	
	\draw (2,-4) node[above]{\textbf{(a)}};
\end{scope}

\begin{scope}[shift = {(6,0)}]
	\coordinate(Aa) at (0,0) ;
	\coordinate (Ab) at (0,-3) ;
	\coordinate (Ba) at (1,0) ;
	\coordinate (Bb) at (1,-3) ;
	\coordinate (Ca) at (2,0) ;
	\coordinate (Cb) at (2,-3) ;
	\coordinate (Pa) at (3,0) ;
	\coordinate (Pb) at (3,-3) ;
	\draw (Aa) node[above]{$p$} ;
	\draw (Ba) node[above]{$q$} ;
	\draw (Ca) node[above]{$r$} ;
	\draw (Pa) node[above]{$\pi$} ;
	\draw (Aa) -- (Ab) ;
	\draw (Ba) -- (Bb) ;
	\draw (Ca) -- (Cb) ;
	\draw (Pa) -- (Pb) ;
	\coordinate (sa) at (0,-0.5); 
	\coordinate (ra) at (3,-0.5); 
	\draw[>=latex,->] (sa) to node[above, sloped, pos = 0.17]{$(r,\amessage_1)$} (ra); 
	
	\coordinate (sb) at (0, -1); 
	\coordinate (rb) at (1, -1); 
	\draw[>=latex,->] (sb) -- (rb); 
	\draw (sb) node[above right]{$\amessage_2$};
	
	\coordinate (sc) at (2, -1.5); 
	\coordinate (rc) at (1, -1.5); 
	\draw[>=latex,->] (sc) -- (rc); 
	\draw (sc) node[above left]{$\amessage_3$};

	\coordinate (sd) at (3, -2); 
	\coordinate (rd) at (2, -2); 
	\draw[>=latex,->] (sd) -- (rd); 
	\draw (sd) node[above left]{$\amessage_1$};
	
	\draw (1.5,-4) node[above]{\textbf{(b)}};
\end{scope}

\begin{scope}[shift = {(11,0)}]
	\coordinate(Aa) at (0,0) ;
	\coordinate (Ab) at (0,-3) ;
	\coordinate (Ba) at (1,0) ;
	\coordinate (Bb) at (1,-3) ;
	\coordinate (Ca) at (2,0) ;
	\coordinate (Cb) at (2,-3) ;
	\coordinate (Pa) at (3,0) ;
	\coordinate (Pb) at (3,-3) ;
	\draw (Aa) node[above]{$p$} ;
	\draw (Ba) node[above]{$q$} ;
	\draw (Ca) node[above]{$r$} ;
	\draw (Pa) node[above]{$\pi$} ;
	\draw (Aa) -- (Ab) ;
	\draw (Ba) -- (Bb) ;
	\draw (Ca) -- (Cb) ;
	\draw (Pa) -- (Pb) ;
	\coordinate (sa) at (0,-0.5); 
	\coordinate (ra) at (3,-0.5); 
	\draw[>=latex,->] (sa) to node[above, sloped, pos = 0.17]{$(r,\amessage_1)$} (ra); 
	
	\coordinate (sb) at (0, -1); 
	\coordinate (rb) at (1, -1); 
	\draw[>=latex,->] (sb) -- (rb); 
	\draw (sb) node[above right]{$\amessage_2$};
	
	\coordinate (sc) at (1, -1.5); 
	\coordinate (rc) at (0, -1.5); 
	\draw[>=latex,->] (sc) to node[above,pos=0.8]{$\amessage_3$} (0.65,-1.5); 
	\draw[>=latex, ->, dashed] (0.65,-1.5)--(rc); 
	
	\coordinate (sd) at (2, -2); 
	\coordinate (rd) at (1, -2); 
	\draw[>=latex,->] (sd) -- (rd); 
	\draw (sd) node[above left]{$\amessage_4$};

	\coordinate (se) at (3, -2.5); 
	\coordinate (re) at (2, -2.5); 
	\draw[>=latex,->] (se) -- (re); 
	\draw (se) node[above left]{$\amessage_1$};
	
		\draw (1.5,-4) node[above]{\textbf{(c)}};

\end{scope}
\end{tikzpicture}
\caption{MSCs of problematic executions}\label{fig:problematic-exe}
\end{center}

\end{figure}

\paragraph{\bf Other differences with \cite{DBLP:conf/cav/BouajjaniEJQ18}. }
Our definition of causal delivery slightly differs from the one in \cite{DBLP:conf/cav/BouajjaniEJQ18}. Indeed our Property \ref{prop:cau}  does not hold for the
definition in \cite{DBLP:conf/cav/BouajjaniEJQ18}.  The two  examples  below  stress where the definition of causal delivery in \cite{DBLP:conf/cav/BouajjaniEJQ18} fails.
Nonetheless here we have merely fixed a typo, as in the subsequent development in  \cite{DBLP:conf/cav/BouajjaniEJQ18}, causal delivery is used as intended by our definition.

\begin{example}\label{ex:counter-causal}
Let $e_1$ be a sequence of actions such that its $msc(e_1)$ is the one depicted in Fig.~\ref{fig:counter-exmp-causal}a. According to the definition in \cite{DBLP:conf/cav/BouajjaniEJQ18}, causal delivery is satisfied when pairs  of message exchanges with 
identical receivers have  sends that are causally related. Thus we have:
\begin{itemize}
\item $\send{p}{q}{\amessage_1} < \send{p}{r}{\amessage_2} < \rec{p}{r}{\amessage_2} < \send{r}{q}{\amessage_3} \in po \cup src$ and $\rec{r}{q}{\amessage_3} < \rec{p}{q}{\amessage_1} \notin po$ 
\item $ \send{r}{s}{\amessage_4} < \send{r}{t}{\amessage_5} < \rec{r}{t}{\amessage_5} < \send{t}{s}{\amessage} \in po \cup src$ and $\rec{t}{s}{\amessage_6} < \rec{r}{s}{\amessage_4} \notin po$
\end{itemize}   
 This entails that $msc(e_1)$ satisfies causal delivery. However, there is no  execution corresponding to this MSC as it is impossible to find a linearisation of  $msc(e_1)$. In our Definition \ref{def:causal-delivery}, instead, we add the requirement that a linearisation of $msc(e_1)$ must exist. Thus $msc(e_1)$ does not satisfy causal delivery.
\end{example}

\begin{example}\label{ex:counter-causal2}
Let $e_2$ be a sequence of actions and $msc(e_2)$ as depicted in Fig.~\ref{fig:counter-exmp-causal}b. As in the previous example, according to the definition in \cite{DBLP:conf/cav/BouajjaniEJQ18}, $msc(e_2)$ satisfies causal delivery, indeed we  check only messages with identical receiver and whose sends are causally dependent: 
\begin{itemize}
\item $\send{p}{q}{\amessage_3} < \send{p}{r}{\amessage_4} < \rec{p}{r}{\amessage_4} < \send{r}{q}{\amessage_5} \in po \cup src$ and $\rec{r}{q}{\amessage_5} < \rec{p}{q}{\amessage_3} \notin po$
\end{itemize}
However, it ignores the dependency between $\amessage_2$ and $\amessage_3$. Indeed, the mailbox communication implies that if $\rec{t}{q}{\amessage_2} < \rec{p}{q}{\amessage_3}$ with the same receiver then $\send{t}{q}{\amessage_2} < \send{p}{q}{\amessage_3}$. With this additional constraint, it is impossible to find a linearisation. We can deduce that the definition of causal delivery in \cite{DBLP:conf/cav/BouajjaniEJQ18} is not complete  and should also consider the order imposed by the mailbox communication. In this case, a causally dependency would have been detected between $\amessage_1$ and $\amessage_6$ and we would have seen that the receptions do not happen in the correct order. 
\end{example}
\begin{figure}[t]
\begin{center}
\begin{tikzpicture}[scale = 0.7]
\begin{scope}
\draw (2,-4.8) node[above] {$\textbf{(a)}$};
\draw (0,0.2) node[above]{$p$} ;
\draw (1,0.2) node[above]{$q$}; 
\draw (2,0.2) node[above]{$r$};
\draw (3,0.2) node[above]{$s$};
\draw (4,0.2) node[above]{$t$};
\draw (0,0.2) edge (0,-3.8); 
\draw (1,0.2) edge (1,-3.8);
\draw (2,0.2) edge (2,-3.8); 
\draw (3,0.2) edge (3,-3.8); 
\draw (4,0.2) edge (4,-3.8);
\draw[>=latex, ->] (0,-0.5) to node [pos = 0.7, above, scale = 0.9] {$\amessage_1$} (1,-0.5);
\draw[>=latex, ->] (0,-1) to node [above, sloped, pos = 0.2, scale = 0.9] {$\amessage_2$} (2,-1);
\draw[>=latex, ->] (2,-1.5) to node [above, scale = 0.9, pos=0.3] {$\amessage_3$} (1,-1.5);
\draw[>=latex, ->] (2,-2) to node [above, scale = 0.9] {$\amessage_4$} (3,-2);
\draw[>=latex, ->] (2,-2.5) to node [above, scale = 0.9,pos=0.7] {$\amessage_5$} (4,-2.5);
\draw[>=latex, ->] (4,-3) to node [above, scale = 0.9] {$\amessage_6$} (3,-3);
\draw[>=latex, ->] (4,-3.5) to node [below, pos = 0.1, scale = 0.9] {$\amessage_7$} (0,0);
\end{scope} 

\begin{scope}[xshift = 6.5cm]
\draw (2,-4.8) node[above] {$\textbf{(b)}$};
\draw (0,0.2) node[above]{$p$} ;
\draw (1,0.2) node[above]{$q$}; 
\draw (2,0.2) node[above]{$r$};
\draw (3,0.2) node[above]{$s$};
\draw (4,0.2) node[above]{$t$};
\draw (0,0.2) edge (0,-3.8); 
\draw (1,0.2) edge (1,-3.8);
\draw (2,0.2) edge (2,-3.8); 
\draw (3,0.2) edge (3,-3.8); 
\draw (4,0.2) edge (4,-3.8);
\draw[>=latex, ->] (4,-0.5) to node [pos = 0, right, scale = 0.9] {$\amessage_1$} (3,-3.5);
\draw[>=latex, ->] (4,-1) to node [above, sloped, pos = 0.5, scale = 0.9] {$\amessage_2$} (1,-1);
\draw[>=latex, ->] (0,-1.5) to node [above, scale = 0.9, pos=0.5] {$\amessage_3$} (1,-1.5);
\draw[>=latex, ->] (0,-2) to node [above, scale = 0.9, pos=0.25] {$\amessage_4$} (2,-2);
\draw[>=latex, ->] (2,-2.5) to node [above, scale = 0.9,pos=0.5] {$\amessage_5$} (1,-2.5);
\draw[>=latex, ->] (2,-3) to node [above, scale = 0.9] {$\amessage_6$} (3,-3);
\end{scope} 

\end{tikzpicture}
\end{center}
\caption{MSCs violating
causal delivery}
\label{fig:counter-exmp-causal}
\end{figure}
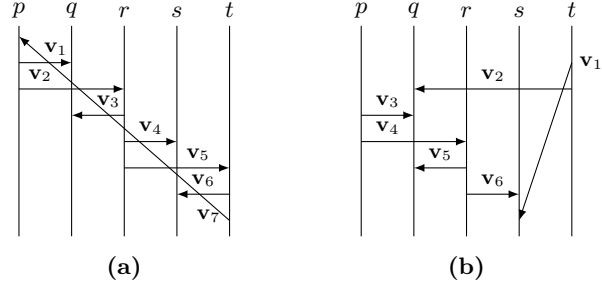

%
%

\section{Additional material}
\label{sec:additional}

\begin{definition}[Instrumented system $\system'$]\label{def:instrumented-system}
  Let $\system = ((L_p, \delta_p, l_p^0) \mid p \in \procSet)$ be a system of communicating machines. The instrumented system $\system'$ associated to $\system$
  is defined such that $\system'= \left( (L_p, \delta'_p, l_p^0) \mid p \in \procSet \cup \{\pi\} \right ) $ where
  for all $p \in \procSet$: 
  \begin{align*}
  \delta'_p =  \delta_p & \cup \{l_1 \xrightarrow{\send{p}{\pi}{(q,\amessage)}} l_2 \mid l_1 \xrightarrow{\send{p}{q}{\amessage}} l_2 \in \delta_p \} \\
  & \cup \{l_1 \xrightarrow{\rec{\pi}{p}{\amessage}} l_2 \mid l_1 \xrightarrow{\rec{q}{p}{\amessage}} l_2 \in \delta_p \}
  \end{align*} 
Process $\pi$ is the communicating automaton
  $(L_\pi, l_{\pi}^0, \delta_\pi )$ 
  where
  \begin{itemize}
  \item
    $L_{\pi} = \{l_\pi^0, l_f\} \cup \{l_{q,\amessage} \mid \amessage \in \paylodSet, q\in\procSet{} \} $, 
    and
  \item 
$\delta_{\pi} = \{l_\pi^0 \xrightarrow{\rec{p}{\pi}{(q,\amessage)}} l_{q,\amessage} \mid \send{p}{q}{\amessage} \in \sendSet \}$ \\ 
\hspace*{1cm}$   \cup \{l_{q,\amessage} \xrightarrow{\send{\pi}{q}{\amessage}} l_f \mid \rec{p}{q}{\amessage} \in \receiveSet\}
 $
  \end{itemize}
\end{definition}

\section{Proofs of Lemmata and Theorems}

\subsection*{Proof of Theorem~\ref{thm:k-sity-scc}}

\subsection*{Proof of Theorem~\ref{thm:causal-delivery-graphically}}

\subsection*{Proof of Lemma~\ref{lem:causal}}

\subsection*{Proof of Theorem \ref{thm:reachability-is-decidable}}
\begin{quote}
{\bf Theorem 3}
  Let $\system$ be a \kSable{k} system and $\globalstate{l}$ a global
  control state of $\system$. The problem whether there
  exists $e\in\asEx(\system)$ and $\B$ such that
  $(\globalstate{l_0},\B_0)\xRightarrow{e}(\globalstate{l},\B)$
  is decidable.
\end{quote}

\begin{proof}
There are only finitely many abstract configurations of the form
  $(\globalstate{l},B)$ with $\globalstate{l}$ a tuple of control states and
  $B:\procSet \to (2^{\procSet}\times 2^{\procSet})$.
  Therefore $\transitionKE{e}{k}$ is a relation
  on a finite set, and the set $\sTr_k(\system)$ of \kSous{k} MSCs of a
  system $\system$ forms a regular language. It follows that,  it is decidable
  whether a given abstract configuration of the form $(\globalstate{l}, B)$ is reachable from the initial configuration following a 
  \kSable{k} execution being a linearisation of an MSC contained in $sTr_k(\system)$. \qed
  
\end{proof}

\subsection*{Proof of Lemma~\ref{lem:bad-is-bad}}

\begin{quote}
{\bf Lemma 2}
  A system $\system$ is not \kSable{k} iff there is a \kSable{k}
  execution $e'$\! of $\system'$ that is feasible and bad.
\end{quote}

\begin{proof}
$\Rightarrow$
Let $\system$ be not \kSable{k} then there exists an execution that is not \kSable{k} which contains a unique minimal prefix of the form $e \cdot r$ with $e$ \kSable{k} and $r = \rec{p}{q}{\amessage}$ a receive action. Thus $e$ is bad and there exists an $e' = \deviate{e\cdot r} \in \asEx(\system')$.

Since $e$ is \kSable{k}, 
$msc(e)$ is \kSous{k} and there exists a linearisation $e''$ such that 
$e''=e_1\dots e_n$ and there exists a \kE{k} $e_i$ containing the send action  $\send{p}{q}{\amessage}$. Now if we replace this action with $\send{p}{\pi}{(q,\amessage)}$ and we add at the end of the same \kE{k} the action $\rec{p}{\pi}{(q,\amessage)}$. The execution in $\asEx(\system')$ remains \kSable{k}. 
Finally if we add to $e''$ a new \kE{k} with the actions $\send{\pi}{q}{\amessage}$ and $\rec{\pi}{q}{\amessage}$ the execution remains \kSable{k}, hence $e'$ is feasible.

$\Leftarrow$
If there is a \kSable{k}
  execution $e'$ of $\system'$ that is feasible and bad, then by construction $e' = \deviate{e\cdot r}$ and $e\cdot r$ is not \kSable{k}. Whence $\system$ is not \kSable{k} and this concludes the proof. \qed
  
  \end{proof}

\subsection*{Proof of Lemma~\ref{lem:graph-carac-feas}}

\begin{quote}
{\bf Lemma 3}
  The execution $e'$ is not feasible iff
  there is a vertex $v$ in
      the conflict graph of $e'$ such that
      $\vstart\cgedgeD{SS} v\cgedge{RR}\vstop$.
\end{quote}

\begin{proof}

$\Leftarrow$ 
If there is $v$ such that $\vstart\cgedgeD{SS}v\cgedge{RR}\vstop$,
  this means that a message sent after the deviated message is received before
  it: hence, it violates causal delivery.

  $\Rightarrow$ Assume now $e'$ not feasible this entails that $e$ is \kSable{k} and $e\cdot r$ violates causal
  delivery. Whence there exists an
  unmatched message that becomes matched and because of  Definition~\ref{def:causal-delivery}  
  there are $i',j'$ such that $r=a_{i'}$ $a_i\matches a_{i'}$, $a_j\matches a_{j'}$, and $i \prec j$ and $j'\prec i'$.
  So the conflict graph $\cgraph{e\cdot r}$ contains two vertices
  $v_d=\{a_i, a_{i'}\}$ and $v\{a_j, a_{j'}\}$ such that $v_d\cgedgeD{SS}v\cgedge{RR}v_d$. Because of the deviation  node $v_d$ is split in nodes $\vstart$
  and $\vstop$ in $\cgraph{e'}$, and therefore we conclude that 
  $\vstart\cgedgeD{SS}v\cgedge{RR}\vstop$. \qed
\end{proof}

\subsection*{Proof of Lemma~\ref{lem:feasible-is-regular}}

\subsection*{Proof of Lemma~\ref{lem:bad-characterization}}
\begin{quote}
{\bf Lemma 5}
The feasible execution $e'$ is bad iff  one of the two
holds
\begin{enumerate}
\item 
 $\vstart\cgedge{}^*\cgedge{RS}\cgedge{}^*\vstop$, or 
\item 
the size of the set $\succs{\vstart}\cap\preds{\vstop}$ is greater or equal to $k+2$.
\end{enumerate}
\end{quote}

\begin{proof}

Since $msc(e')$ is \kSous{k} and $e' = \deviate{e\cdot r}$, $msc(e)$ (without the last reception $r$)
is \kSous{k}.
By Theorem~\ref{thm:k-sity-scc}, $e'$ is bad if and only if
$\cgraph{e\cdot r}$ contains either a cyclic path
with an RS edge, or a SCC with of size $\geq k+1$.
This cyclic path (resp. SCC) must contain
the vertex associated with the last receive $r$ of $e\cdot r$. In $\cgraph{e'}$, this cyclic (resp. SCC) corresponds to a path
from $\vstart$ to $\vstop$ (resp. the set of vertices that
are both reachable from $\vstart$ and co-reachable from $\vstop$). Since
the $\vstart$ and $\vstop$ account for the same node in the conflict graph
of $e\cdot r$, the size of the SCC is one less than the size of the set
$\succs{\vstart}\cap\preds{\vstop}$. \qed
\end{proof}

\subsection*{Proof of Lemma~\ref{lem:transitionBV}}

\subsection*{Proof of Theorem~\ref{thm:ksity-decidability}}
\begin{quote}
{\bf Theorem 4}
The \kSity{k} of a system $\system$ is decidable for $k \geq 1$. 
\end{quote} 
 
 \begin{proof}
  Let $\system$ be fixed.
  By Lemmata~\ref{lem:bad-is-bad}, \ref{lem:feasible-is-regular}, and \ref{lem:transitionBV},
  $\system$ is not \kSable{k} if and only if there is a sequence of actions
  $e'=e_1'\cdots e_{n}'\cdot s\cdot{}r$
  such that $e_i\in S^{\leq k}R^{\leq k}$,   $s=\send{\pi}{q}{\amessage}$, $r=\rec{\pi}{q}{\amessage}$,
  $$
  (\globalstate{l_0}, B_0, (\emptyset,\emptyset), \bot)\transitionCD{e_1'}{k}\dots\transitionCD{e_{n}'}{k}
  (\globalstate{l'}, B', \vec C',q)
  $$
  and
  $$
  (\{\pi\},Q,\ctfalse,0)
  \transitionBV{e_1'}{k}\dots\transitionBV{e_{n}'}{k}\transitionBV{s\cdot{}r}{k}
  (P',\{\pi\},\sawRS,\cont)
  $$
  for some $\globalstate{l'},B',\vec C',Q,P'$ with $\pi\not\in C_{R,q}$.
  Since both relations $\transitionCD{e}{k}$ and $\transitionBV{e}{k}$ are finite state, the existence of such a sequence of actions is decidable. \qed
 \end{proof}

\section{Additional material for peer-to-peer systems}\label{sec:additional_p2p}


Note that all those notions that have not been redefined have the same definition as in the mailbox setting. 

 \begin{definition}[Peer-to-peer configuration]
Let $ \system = \left( (L_p, \delta_p, l^0_p) \mid p \in \procSet \right)$, a configuration is a pair $(\globalstate{l}, \B)$ where $\globalstate{l} = (l_p)_{p\in \procSet}  \in  \Pi_{p \in \procSet} L_p$ is a global control state of $\system$ (a local control state for each automaton) and $\B = (b_{pq})_{p,q\in\procSet} \in (\paylodSet^*)^\procSet$ is a vector of buffers, each $b_{pq}$ being a word over $\paylodSet$.   
\end{definition}

\begin{definition}[Peer-to-peer semantics]
\begin{center}
\begin{prooftree}
\AxiomC{\small$ \Transition{l_p}
               {\send{p}{q}{\amessage}}
               {l_p'}
               {p}
    \quad
    b_{pq}' = b_{pq} \cdot \amessage$}
 \LeftLabel{[SEND]}
\UnaryInfC{\small$\Transition{(\globalstate{l}, \B)}
               {\send{p}{q}{\amessage}}
	       {(\globalstate{l}\sub{l_p'}{l_p},\B\sub{b_{pq}'}{b_{pq}})}
               {}$}
\end{prooftree}
\begin{prooftree}
\AxiomC{\small$ \Transition{l_q}
               {\rec{p}{q}{\amessage}}
               {l_q'}
               {q}
    \quad
    b_{pq} = \amessage \cdot b_{pq}'$}
        \LeftLabel{[RECEIVE]}
\UnaryInfC{\small$\Transition{(\globalstate{l}, \B)}
               {\rec{p}{q}{\amessage}}
	       {(\globalstate{l}\sub{l_q'}{l_q},\B\sub{b_{pq}'}{b_{pq}})}
               {}$}
\end{prooftree}
\end{center}
\end{definition}

\begin{definition}[Peer-to-peer causal delivery]
Let $msc = (Ev , \lambda, \prec)$ be an MSC. 
  We say that $msc$ satisfies causal delivery if
  there is a linearisation $e~=~a_1\dots a_n$ such that
  for any two send events $i \prec j$ such that $a_i~=~\send{p}{q}{\amessage}$ and $a_j = \send{p'}{q'}{\amessage'}$, $p=p'$ and $q = q'$,  either $a_j$ is unmatched, or there are $i',j'$ such that $a_i\matches a_{i'}$, $a_j\matches a_{j'}$, and $i' \prec j'$.

\end{definition}

\begin{theorem}
Let $e$ be a sequence of actions such that $msc(e)$ satisfies causal delivery. Then $msc(e)$ is \kSous{k} iff every SCC in its conflict graph is of size at most $k$ and if no RS edge occurs on any cycle path. 
\end{theorem}
\begin{proof}
  Analogous to the proof of Theorem~\ref{thm:k-sity-scc}. \qed
 \end{proof}

\begin{figure}[t]
\begin{prooftree}
\AxiomC{
	\stackanchor
		{\stackanchor
			{$ e = s_1 \cdots s_m\cdot r_1\cdots r_{m'} \qquad s_1 \cdots s_m  \in \sendSet^* \qquad r_1 \cdots r_{m'} \in \receiveSet^* \qquad 0 \leq m' \leq m \leq k$ 
			}
			{$(\globalstate{l}, \B_0) \xRightarrow{e} (\globalstate{l'} , \B)$ for some $ \globalstate{l'}$ and $\B$
			}
		}
		{\stackanchor
			{\stackanchor{for all $q \in \procSet \quad B_{i+1}(q) = B_i(q) \cup \{ p \mid s_i = \send{p}{q}{\amessage} \mbox{ \& } s_i \mbox{ is unmatched}	\}$ }		
			{
			}
			}
			{for all send action $s_i \in e$, $s_i \matches r_j \implies  \procofaction{}{s_i} \notin B_i(\procofaction{}{r_j})$	
			}
		}
	}
\UnaryInfC{$(l, B) \ptransitionKE{e}{k} (l', B_{m+1}) $}
\end{prooftree}
\caption{Definition of transition $\ptransitionKE{e}{k}$ in a peer-to-peer system}\label{fig:transitionCDp2p}
 \end{figure}

\begin{lemma}
 An MSC $msc$ is \kSous{k} iff
  there is a linearisation  
  $e~=~e_1\cdots e_n$ such that
  $(\globalstate{l_0},B_0)\ptransitionKE{e_1}{k}\cdots\ptransitionKE{e_n}{k}(\globalstate{l'},B')$ for some global state $\globalstate{l'}$ and
  some $B':\procSet\to(2^{\procSet}\times 2^{\procSet})$.
%
\end{lemma}

\begin{proof}
$\Rightarrow$ 
Since $msc$ is \kSous{k} then $\exists e = e_1 \cdots e_n $ such that $e$ is a linearisation of $msc$. The proof proceeds by induction on $n$.  

\begin{description}
\item[Base case] If $n=1$ then $e=e_1$. Thus there is only one \kE{k}. By hypothesis, as $msc$ satisfies causal delivery we have for some $\globalstate{l}$ and $\B$, $(\globalstate{l}, \B_0) \xRightarrow{e} (\globalstate{l}, \B)$.
By contradiction, suppose that $\exists v = \{s_i, r_{i'}\}$ such that $s_i = \send{p}{q}{\amessage}$ and $ p \in B_i(q)$. 
Then $\exists v' = \{s_j\}$ such that $s_j = \send{p}{q}{\amessage'}$.   
Since $e$ is a linearisation of $msc$ and $\procofactionv{S}{v} = \procofaction{S}{v'}$, then $j \prec_{po} i$. As $\procofactionv{R}{v'} =\procofactionv{R}{v}$ and $v$ is matched while $v'$ is not matched, $msc$ does not satisfy causal delivery which is a contradiction.  

\item[Inductive step] If $n>1$, by inductive hypothesis, we have $$(\globalstate{l_0}, B_0) \ptransitionKE{e_1}{k} \cdots \ptransitionKE{e_{n-1}}{k} (\globalstate{l_{n-1}}, B)$$
Since all receptions match a corresponding send in the current \kE{k} and all sends precede all receptions, we have $(\globalstate{l_{n-1}}, \B_0) \xRightarrow{e} (\globalstate{l_n}, \B)$ for some $\B$.  

Inductive hypothesis entails that 
$$ B(q) =  \{ \procofactionv{S}{v} \mid v \mbox{ is unmatched \& } \procofactionv{R}{v} = q \} $$


By contradiction, we suppose that $\exists v = \{s_i, r_{i'}\}$ such that $s_i = \send{p}{q}{\amessage}$ and $ p \in B_i(q)$. 
Then there exists a message exchange $v'$  in $e$ such that 
	$v' = \{s_j \}$, $s_j = \send{p}{q}{\amessage'} $. 
 As $v'$ in $e$, and as in the base case, $j \prec_{po} i$ with $v$ matched and $v'$ unmatched, $msc$ would  not satisfy causal delivery which is a contradiction
 

\end{description}

$\Leftarrow$ 
If $e=e_1\cdots e_n$ where each $e_i$ corresponds to a valid \kE{k}.
Let show that $msc(e)$ is \kSous{k}.  

Suppose by contradiction that $msc(e)$ is not \kSous{k}. As $e$ is a linearisation of $msc(e)$ and is divisible into valid \kE{k}, then $msc(e)$ violates causal delivery. 
Then there exist $ s_i = \send{p}{q}{\amessage}, s_j=\send{p}{q}{\amessage'}$ such that $i \prec j$ and either: 
\begin{itemize}
\item there exist $r_{i'} = \rec{p}{q}{\amessage}, r_{j'}=\rec{p}{q}{\amessage'}$ such that $j' \prec i'$ or,
\item $s_i$ is unmatched and $s_j$ is matched
\end{itemize}
\noindent Moreover, we have that either $s_i,s_j \in e_l$ or $s_i \in e_l$ and $s_j \in e_m$ with $l\neq m$.  
We thus have four cases.

In the first case, $s_i$ and $s_j'$ are matched and belong to the same \kE{k} $e_l$. Then there is no valid execution  such that $(\globalstate{l}, \B_0) \xRightarrow{e_l} (\globalstate{l'}, \B)$ and $e_l$ does not describes a valid \kE{k}. 
In the second case, $s_i$ and $s_j$ are matched but do not belong to the same \kE{k}. This case cannot happen as  $j' \prec i'$ entails that $s_i$ and $s_j$ must belong to the same \kE{k} or the reception of message $s_i$ will be separated from its sending.  
In the third case, $s_i$ is unmatched and $s_j$ is matched, and they are in the same \kE{k}. If $s_i $ is unmatched, then $p \in B_{i+1}(q)$. Moreover as $i \prec j$ ,  $p \in B_j(q)$ thus concluding that  $e_l$ is not a valid \kE{k}. 
In the last case, $s_i \in e_l$ is unmatched and $s_j \in e_{m}$ is matched. Therefore, since function $B$ is incremental, $p \in B(q)$ at the beginning of $e_{m}$. Then,  $p \in B_j(q)$ and $e_{m}$ is not a valid \kE{k} concluding the proof. \qed
%
%
%
%
\end{proof}

 \begin{lemma}\label{lem:bad-is-bad-p}
  A system $\system$ is not \kSable{k} iff there is a \kSable{k}
  execution $e'$ of $\system'$ that is feasible and bad.
\end{lemma}
\begin{proof}
 Analogous to the proof of Lemma~\ref{lem:bad-is-bad}. \qed
  \end{proof}

\begin{figure}[t]
\begin{prooftree}
\AxiomC{
	\stackanchor
		{\stackanchor
			{$(\globalstate{l}, B) \ptransitionKE{k}{e} (\globalstate{l'}, B') \qquad e= a_1\cdots a_n \qquad (\forall v) \procofactionv{S}{v} \neq \pi$
			}
			{
			$(\forall v,v') \procofactionv{R}{v} = \procofactionv{R} {v'} = \pi \implies v = v' \wedge \finddest = \bot$ 
			}
		}
		{
		\stackanchor
			{\stackanchor
				{
				$(\forall i) a_i = \send{p}{\pi}{(q,\amessage)} \implies \finddest' = q \wedge d = i$
				}				
				{ 
				$\finddest \neq \bot \implies \finddest' = \finddest$
				}
			}
			{\stackanchor
				{\stackanchor
					{			
					$ \finddest' \neq \bot \wedge \finddest = \bot \implies \nexists v (v \mbox{ is matched } \wedge s_j \in v \wedge j > d $
					}
					{ 
					$\wedge \procofactionv{S}{v} = \findexp \wedge \procofactionv{R}{v} = \finddest) $	
					}
			 	}
			 	{
			 	$ \finddest' \neq \bot \wedge \finddest \neq \bot \implies \nexists v ( v \mbox{ is matched } \wedge \procofactionv{S}{v} = \findexp \wedge \procofactionv{R}{v} = \finddest) $	
			 	}
			}
		}
	}
\UnaryInfC{$(\globalstate{l}, B, \findexp, \finddest) \ptransitionCD{e}{k} (\globalstate{l'}, B', \findexp', \finddest') $}
\end{prooftree}
\caption{Definition of transition $\ptransitionCD{e}{k}$ in a peer-to-peer system }\label{fig:transitionfeasp2p}
 \end{figure}

\subsection*{Proof of Theorem \ref{thm:reachability-dec-p2p}}
\begin{quote}
{\bf Theorem 5}
Let $\system$ be a \kSable{k} system and $\globalstate{l}$ a global control state of $\system$. The problem whether there exists $e \in asEx(\system)$ and $\B$ such that $(\globalstate{l_0}, \B_0) \xRightarrow{e} (\globalstate{l}, \B)$ is decidable.
\end{quote}

\begin{proof}
There are only finitely many abstract configurations of the form
  $(\globalstate{l},B)$ with $\globalstate{l}$ a tuple of control states and
  $B:\procSet \to (2^{\procSet})$.
  Therefore $\ptransitionKE{e}{k}$ is a relation
  on a finite set, and the set $\sTr_k(\system)$ of \kSous{k} MSCs of a
  system $\system$ forms a regular language. It follows that,  it is decidable
  whether a given abstract configuration of the form $(\globalstate{l}, B)$ is reachable from the initial configuration following a 
  \kSable{k} execution being a linearisation of an MSC contained in $sTr_k(\system)$. \qed \end{proof}

\begin{lemma}\label{lem:feasible-is-regular-p}
Let $e'$ be an execution of $\system'$. 
Then $e'$ is a \kSable{k} feasible execution iff there are 
	$e'' = e_1\cdots e_n\cdot \send{\pi}{q}{\amessage}\cdot \rec{\pi}{q}{\amessage}$ with $e_1,\ldots ,e_n\in S^{\leq k}R^{\leq k}$ such that $msc(e') = msc(e'')$ , $B':\procSet\to2^{\procSet}$, a tuple of control states $\globalstate{l'}$ and processes $p$ and $q$ such that   
$$
(\globalstate{l_0}, B_0, \bot, \bot)
\ptransitionCD{e_1}{k}\dots\ptransitionCD{e_n}{k}
(\globalstate{l'}, B', p, q).
$$
\end{lemma}

\begin{proof}
$\Rightarrow$
Let $e'$ be a \kSable{k} feasible execution of $\system'$. We show that there exists $e'' = e_1\cdots e_n\cdot \send{\pi}{q}{\amessage}\cdot \rec{\pi}{q}{\amessage}$ with $e_1,\ldots ,e_n\in S^{\leq k}R^{\leq k}$, such that $msc(e') = msc(e'')$ , $B':\procSet\to2^{\procSet}$, a tuple of control states $\globalstate{l'}$ and processes $p$ and $q$ such that  
$$
(\globalstate{l_0}, B_0, \bot, \bot)
\ptransitionCD{e_1}{k}\dots\ptransitionCD{e_n}{k}
(\globalstate{l'}, B', p, q).
$$. 

Since $e'$ is \kSable{k}, $msc(e')$ is  \kSous{k} and it exists $e''$ such that $msc(e'') = msc(e')$ and $e'' = e_1 \cdots e_n$ where each $e_i$ is a valid \kE{k}. Therefore, there exist $\globalstate{l'}$ and $B'$ such that  $(\globalstate{l}, B) \ptransitionKE{e_1}{k} \cdots \ptransitionKE{e_n}{k} (\globalstate{l'},B')$. 

Since $e'$ is feasible, let $e\cdot r \in asEx(\system)$ such that $\deviate{e \cdot r} = e'$. So there is a send to $\pi$ in $e'$ and $e' = e'_1 \cdot \send{p}{\pi}{(q,\amessage)} \cdot \rec{p}{\pi}{(q,\amessage)} \cdot e'_2 \cdot  \send{\pi}{q}{\amessage} \cdot \rec{\pi}{q}{\amessage}$. Then there is one and only one send to $\pi$ such that $\findexp = p$ and $\finddest = q$. 

By contradiction, suppose that there is a matched message $v'=\{a_i, a_j\}$ belonging to a \kE{k} in $e'_2$ such that 
$\procofactionv{S}{v'} = p$ and  $\procofactionv{R}{v'} = q$. 
Now if we consider the non-deviated sequence $e\cdot r$ let $a_{i'} = \send{p}{q}{\amessage}$ and $r =a_{j'} = \rec{p}{q}{\amessage}$. We, thus have $i' \prec i $ and $j \prec j'$ contradicting causal delivery and the fact that $e\cdot r$ is an execution.


$\Leftarrow$
Take a sequence of actions $e''$  such that $$(\globalstate{l_0}, B_0, \bot, \bot) \ptransitionCD{e_1}{k} \cdots \ptransitionCD{e_n}{k} (\globalstate{l'}, B', p, q)$$ and $e'$ an execution of $\system'$ such that $msc(e'') = msc(e')$.  

By contradiction, suppose that $e'$ is not feasible. Then $e' = \deviate{e \cdot r}$ with $e \cdot r$ that is not an execution of $\system$ whence it does not satisfy causal delivery. 
If $e \cdot r$ does not satisfy causal delivery, but $e'$ does then $\exists v =\{s_i, r_i'\}, v'= \{s_j, r_j'\}$ such that
$s_i = \send{p}{q}{\amessage}, r_{i'} = \rec{p}{q}{\amessage} $ and $s_j = \send{p}{q}{\amessage'}, r_{j'} = \rec{p}{q}{\amessage'} $ and in $msc(e \cdot r)$: $i\prec j$ and $j' \prec i'$. 
In $e'$, there exists an action $s_l = \send{p}{\pi}{(q,\amessage)}$ 
such that in $msc(e')$: $l \prec j $. Also, the send from $\pi$ being the last action and let $\rec{p}{q}{\amessage'} = r_m$ then $j' \prec m$.  Then, there exist a \kE{k} after the deviation of message $\amessage$ where the exchange $v'$ appears. Thus we have $\finddest' \neq \bot$, $\procofactionv{S}{v'} = p$ and $\procofactionv{R}{v'} = q$.
Moreover, if $s_l $ belongs to the \kE{k} under analysis then $l<j$.  This entails that transition $\ptransitionCD{~}{k}$ does not hold, reaching a contradiction. So $e'$ must be a \kSable{k} feasible execution of $\system'$. \qed

\end{proof}

\begin{lemma}
The feasible execution $e'$ is bad iff  one of the two
holds
\begin{itemize}
\item $\vstart\cgedge{}^*\cgedge{RS}\cgedge{}^*\vstop$, or
\item the size of the set $\succs{\vstart}\cap\preds{\vstop}$ is greater or equal to $k+2$.
\end{itemize}
\end{lemma}

\begin{proof}
Analogous to the proof of Lemma~\ref{lem:bad-characterization}. \qed
\end{proof}

\begin{figure}[t]
 \begin{prooftree}
\AxiomC {\small
	\stackanchor{ $P'=\procs{\succse{e}{P}} \qquad Q=\procs{\predse{e}{Q'}} \qquad SCC_e = \succse{e}{P}\cap\predse{e}{Q' }$}
		{ \stackanchor {$\cont' = \mathsf{min}(k+2,\cont + n) \quad \mbox{where } n= | SCC_e |$}
			{\stackanchor {$
\lastisRec'(q) \Leftrightarrow (\exists v. \procofaction{R}{v}=q\wedge v\cap R\neq\emptyset) \vee (\lastisRec(q) \wedge\! \not \exists v \in V.
\procofaction{S}{v}=q )
$}
				{\stackanchor {$
\sawRS' = \sawRS \vee (\exists v)(\exists p\in \procSet\setminus\{\pi\})\
\procofactionv{S}{v} = p \wedge \lastisRec(p) \wedge p\in P \cap Q 
$}}}}}
\UnaryInfC{ \small $
(P,Q, \cont, \sawRS, \lastisRec)
\ptransitionBV{e}{k}
(P', Q', \cont', \sawRS', \lastisRec')
$}
\end{prooftree}
\vspace*{-0.5cm}
\caption{Definition of the relation $\ptransitionBV{e}{k}$ in a peer-to-peer system
}
 \end{figure}

\begin{lemma}\label{lem:transitionBV-p}
Let $e'$ a feasible \kSable{k} execution of $\system'$. Then $e'$ is a bad execution iff there are $e'' =e_1\cdots e_n\cdot \send{\pi}{q}{\amessage}\cdot \rec{\pi}{q}{\amessage}$ with
  $e_1,\ldots ,e_n\in S^{\leq k}R^{\leq k}$ and $msc(e')=msc(e'')$, $P',Q\subseteq \procSet$, $\sawRS\in\{\cttrue,\ctfalse\}$, $\cont\in\{0,\dots,k+2\}$,
such that
$$
(\{\pi\},Q,0,\ctfalse,\lastisRec_0)
\transitionBV{e_1}{k}\dots\transitionBV{e_n}{k}
(P',\{\pi\},\cont,\sawRS,\lastisRec)
$$
and at least one of the two holds: either $\sawRS=\cttrue$, or $\cont=k+2$.

\end{lemma}

\begin{proof}
Analogous to the proof of Lemma~\ref{lem:transitionBV}. \qed
\end{proof}

\subsection*{Proof of Theorem \ref{thm:ksity-decidability-p2p}}
\begin{quote}
{\bf Theorem 6}
The \kSity{k} of a system $\system$ is decidable for  $k\geq 1 $. 
\end{quote}

\begin{proof}
  Let $\system$ be fixed.
  By Lemmata~\ref{lem:bad-is-bad-p}, \ref{lem:feasible-is-regular-p}, and \ref{lem:transitionBV-p},
  $\system$ is not \kSable{k} if and only if there is a sequence of actions
  $e'=e_1'\cdots e_{n}'\cdot s\cdot{}r$
  such that $e_i\in S^{\leq k}R^{\leq k}$,   $s=\send{\pi}{q}{\amessage}$, $r=\rec{\pi}{q}{\amessage}$,
  $$
  (\globalstate{l_0}, B_0, (\emptyset,\emptyset), \bot)\transitionCD{e_1'}{k}\dots\transitionCD{e_{n}'}{k}
  (\globalstate{l'}, B', \vec C',q)
  $$
  and
  $$
  (\{\pi\},Q,\ctfalse,0)
  \transitionBV{e_1'}{k}\dots\transitionBV{e_{n}'}{k}\transitionBV{s\cdot{}r}{k}
  (P',\{\pi\},\sawRS,\cont)
  $$
  for some $\globalstate{l'},B',\vec C',Q,P'$ with $\pi\not\in C_{R,q}$.
  Since both relations $\transitionCD{e}{k}$ and $\transitionBV{e}{k}$ are finite state, the existence of such a sequence of actions is decidable. \qed \end{proof}

\end{document}